\documentclass[3p,times]{elsarticle}
\usepackage{amsmath,amsfonts,amssymb,amsthm,hyperref}
\usepackage[bottom]{footmisc}
\usepackage{tikz}
  \usetikzlibrary{arrows,automata,shapes}
  \usetikzlibrary{positioning,backgrounds}
  \usetikzlibrary{decorations.pathmorphing}
\usepackage{algorithm} 
\usepackage{enumerate} 
\usepackage{algpseudocode}
\usepackage{nccmath}
\usepackage{multirow}
\usepackage[textwidth=1.9cm,color=green!10,textsize=footnotesize]{todonotes}
\usepackage{caption,subcaption}

\newcommand{\floor}[1]{\left\lfloor #1 \right\rfloor}   
\newcommand{\eps}{\varepsilon}
\newcommand{\preq}{\preccurlyeq}
\newcommand{\uu}{u}
\newcommand{\jjj}{{\floor{\log i}}}
\newcommand{\zav}[1]{\left( #1 \right)}
\newcommand{\abs}[1]{\left| #1 \right|}
\newcommand{\myalpha}[2]{\alpha_{#1,#2}}
\newcommand{\vek}[1]{\overline{#1}}

\allowdisplaybreaks

\DeclareMathOperator{\alp}{\rm alph}
\DeclareMathOperator{\A}{\mathcal{A}}
\DeclareMathOperator{\B}{\mathcal{B}}

\newtheorem{thm}{Theorem}
\newtheorem{lem}[thm]{Lemma}

\newtheorem{cor}[thm]{Corollary}
\newdefinition{example}[thm]{Example}
\newdefinition{oprob}[thm]{Open Problem}

\journal{}

\makeatletter
\def\ps@pprintTitle{%
	\let\@oddhead\@empty
	\let\@evenhead\@empty
	\def\@oddfoot{}%
	\let\@evenfoot\@oddfoot}
\makeatother

\begin{document}
\begin{frontmatter}

\title{On the Height of Towers of Subsequences and Prefixes\tnoteref{note1}}

\tnotetext[note1]{A preliminary version of this work was presented at the MFCS~2014 conference~\cite{mfcs2014}.}

\author[cu]{\v St\v ep\'an Holub\fnref{sh}}
  \ead{holub@karlin.mff.cuni.cz}

\author[tud]{Tom\'{a}\v{s} Masopust\fnref{tm}}
  \ead{tomas.masopust@tu-dresden.de}

\author[inria]{Micha\"el~Thomazo\fnref{mt}}
  \ead{michael.thomazo@inria.fr}

\fntext[sh]{Research supported by the Czech Science Foundation grant number 13-01832S}
\address[cu]{Charles University, Sokolovsk\'a 83, 175 86 Praha, Czech Republic}

\fntext[tm]{Research supported by the by the German Research Foundation (DFG) in Emmy Noether grant KR~4381/1-1 (DIAMOND)}
\address[tud]{Institute of Theoretical Computer Science and Center of Advancing Electronics Dresden (cfaed), TU Dresden, Germany}

\address[inria]{Inria, France}
\fntext[mt]{Research supported by the Alexander von Humboldt Foundation}

\begin{abstract}
  A tower is a sequence of words alternating between two languages in such a way that every word is a subsequence of the following word. The height of the tower is the number of words in the sequence. If there is no infinite tower (a tower of infinite height), then the height of all towers between the languages is bounded. We study upper and lower bounds on the height of maximal finite towers with respect to the size of the NFA (the DFA) representation of the languages. We show that the upper bound is polynomial in the number of states and exponential in the size of the alphabet, and that it is asymptotically tight if the size of the alphabet is fixed. If the alphabet may grow, then, using an alphabet of size approximately the number of states of the automata, the lower bound on the height of towers is exponential with respect to that number. In this case, there is a gap between the lower and upper bound, and the asymptotically optimal bound remains an open problem. Since, in many cases, the constructed towers are sequences of prefixes, we also study towers of prefixes.
\end{abstract}

\begin{keyword}
  Automata \sep languages \sep alternating towers \sep subsequences \sep prefixes \sep upper and lower bounds
  \MSC[2010] 68R05 \sep 68Q45
\end{keyword}

\end{frontmatter}

\section{Introduction}
  A {\em tower\/} between two languages is a sequence of words alternating between the languages in such a way that every word is a subsequence of the following word. The number of words in a tower is the height of the tower. For two regular languages represented by nondeterministic finite automata (NFAs), it is decidable in polynomial time whether there exists an infinite tower, that is, a tower of infinite height~\cite{icalp2013}. As a consequence of a more general result~\cite[Lemma~6]{icalp2013}, the existence of a tower of arbitrary height implies the existence of an infinite tower. Therefore, if there is no infinite tower, the height of all towers is bounded.

  The height of maximal finite towers is closely related, for instance, to the complexity of an algorithm computing a piecewise testable separator of two regular languages~\cite{mfcs2014}. Namely, the algorithm requires at least as many steps as is the height of the maximal tower. An interest in piecewise testable separators is in turn motivated by applications in logic on words~\cite{PlaceZ_icalp14,PlaceZ15}, especially in the context of the Straubing-Th\'erien hierarchy~\cite{Straubing81,Therien81}, and by applications in XML Schema languages~\cite{HofmanM15,MartensNNS15}.

  Not much is known about the upper bound on the height of towers between two regular languages if no infinite tower exists. The only result we are aware of is a result by Stern~\cite{Stern-tcs85} giving an exponential upper bound $2^{{|\Sigma|}^2 n}$ on the height of towers between a piecewise testable language over an alphabet $\Sigma$ represented by an $n$-state minimal DFA and its complement. Piecewise testable languages form a strict subclass of regular languages. 
  
  In this paper, we give a better bound that holds in a general setting of two arbitrary regular languages (having no infinite tower) represented by NFAs.
  We show in Theorem~\ref{thm01} that the upper bound on the height of towers between two regular languages represented by NFAs is polynomial with respect to the depth of the NFAs and exponential with respect to the size of the alphabet. 
  
  Considering the lower bound, we first improve in Theorem~\ref{thm02B} an existing bound for binary regular languages~\cite{mfcs2014}.
  Theorem~\ref{thm:lower_bound} and its corollaries then show that the upper bound is asymptotically tight if the alphabet is fixed. 
  If the alphabet may grow with the depth of the automata, Theorem~\ref{thm:2exp:improved} and its corollaries show that we can achieve an exponential lower bound for NFAs with respect to the number of states. Notice that it does not contradict the polynomiality of the upper bound with respect to the number of states because the automata require an alphabet of size approximately the number of states. These lower bounds are not asymptotically equal to the upper bound and it is not known what the (asymptotically) tight bound is, cf. Open Problem~\ref{op1}. Specifically, we do not know whether an alphabet of size grater than the number of states may help to build higher towers.

  Then we discuss the lower bound for a DFA representation of languages. In Theorems~\ref{thm:expdfa}, \ref{determinisation}, and~\ref{determinisation2}, we show that, for any two NFAs, there are two DFAs preserving the height of towers if the size of the alphabet may grow. If the alphabet is fixed, we show in Theorem~\ref{thm:lower_bound:dfa} that the upper bound on the height of towers is asymptotically tight even for DFAs.
  
  The towers we construct to demonstrate lower bounds are mostly sequences of prefixes. Therefore, we also investigate {\em towers of prefixes}. We prove tight bounds on the height of towers of prefixes in Theorem~\ref{thm:dfas} for DFAs and in Theorem~\ref{cor:nfas} for NFAs. We then discuss towers of prefixes between two binary NFAs in Corollary~\ref{cor23}. Finally, we provide a pattern that characterizes the existence of an infinite tower of prefixes in Theorem~\ref{patern}. 

  Our main results are summarized in Table~\ref{tableOverview}.
  We also formulate the following two open problems:
  \begin{enumerate}
    
    \item What is the tight bound on the height of towers of subsequences for two NFAs (DFAs) over an alphabet that may grow with the number of states? See Open Problem~\ref{op1} below.
    
    \item What is the tight bound on the height of towers of prefixes for NFAs over a fixed (binary) alphabet? See Open Problem~\ref{op2}.
    
  \end{enumerate}

  \begin{table*}
    \centering
    \renewcommand{\arraystretch}{1.5}
    \begin{subtable}{0.47\linewidth}
      \centering
      \begin{tabular}[c]{c|c|c|c}
        & Upper bound  & \multicolumn{2}{c}{Lower bound}  \\
        &              & $|\Sigma|=k$  &  $|\Sigma| \ge n_1+n_2$ \\\hline\hline
        NFAs           & \multirow{2}{*}{$\dfrac{n^{|\Sigma|+1}-1}{n-1}$}
                       & \multirow{2}{*}{$\Theta\zav{n^k}$} 
                       & \multirow{2}{*}{$\Omega\zav{2^{n_1+n_2}}$} \\
        DFAs
          & &
      \end{tabular}
      \caption{Towers of sequences over an alphabet $\Sigma$; $n=\max(n_1,n_2)$}
      \label{tab:dimFFT}
    \end{subtable}%
    \hspace{1em}
    \begin{subtable}{0.49\linewidth}
      \centering
        \begin{tabular}[c]{c|c|c|c}
              & Upper bound  & \multicolumn{2}{c}{Lower bound}  \\
              &              & $|\Sigma|=2$  &  $|\Sigma| \ge n_1+n_2$\\\hline\hline
        NFAs  & $\frac{(2^{n_1}-1)(2^{n_2}-1)+1}{2}$     
              & $2^{\Omega\left(\sqrt{\frac{n_1+n_2}{\log(n_1+n_2)}}\right)}$
              & $2^{n_1+n_2-2}$ \\
        DFAs  & $\mfrac{n_1n_2}{2}+1$ 
              & $\mfrac{n_1n_2}{2}+1$ 
              & $\mfrac{n_1n_2}{2}+1$
      \end{tabular}
      \caption{Towers of prefixes}
      \label{tab:dimGMM}
    \end{subtable}
    \caption{Upper and lower bounds on the height of towers of subsequences and prefixes for automata with $n_1$ and $n_2$ states}
    \label{tableOverview}
  \end{table*}

\section{Preliminaries}
  The cardinality of a set $\Sigma$ is denoted by $|\Sigma|$ and the power set of $\Sigma$ by $2^{\Sigma}$. The free monoid generated by $\Sigma$ is denoted by $\Sigma^*$. An element of $\Sigma^*$ is called a {\em word}; the empty word is denoted by $\eps$. For a word $w\in\Sigma^*$, $\alp(w)\subseteq\Sigma$ denotes the set of all letters occurring in $w$, and $|w|_a$ denotes the number of occurrences of letter $a$ in $w$.

  A {\em nondeterministic finite automaton\/} (NFA) is a quintuple $\A = (Q,\Sigma,\delta,Q_0,F)$, where $Q$ is the finite nonempty set of states, $\Sigma$ is the alphabet, $Q_0\subseteq Q$ is the set of initial states, $F\subseteq Q$ is the set of accepting states, and $\delta\colon Q \times \Sigma \to 2^Q$ is the transition function that can be extended to the domain $2^Q \times \Sigma^*$ in the usual way. The language {\em accepted\/} by $\A$ is the set $L(\A) = \{w \in \Sigma^* \mid \delta(Q_0, w) \cap F \neq\emptyset\}$.
  A {\em path\/} $\pi$ from a state $q_0$ to a state $q_n$ under a word $a_1a_2\cdots a_{n}$, for some $n\ge 0$, is a sequence of states and input letters $q_0, a_1, q_1, a_2, \ldots, q_{n-1}, a_{n}, q_n$ such that $q_{i+1} \in \delta(q_i,a_{i+1})$ for all $i=0,1,\ldots,n-1$. The path $\pi$ is {\em accepting\/} if $q_0\in Q_0$ and $q_n\in F$, and it is {\em simple\/} if the states $q_0,q_1,\ldots,q_n$ are pairwise distinct.
  The number of states on the longest simple path in $\A$ is called the {\em depth\/} of $\A$.
  We write $q \xrightarrow{w} q'$ to denote that  $q'\in \delta(q,w)$ and say that there exists a path from state $q$ to state $q'$ under the word $w$, or labeled by the word $w$.
  The NFA $\A$ has a {\em cycle over an alphabet $\Gamma \subseteq \Sigma$\/} if there exists a state $q$ and a word $w$ over $\Sigma$ such that $\alp(w) = \Gamma$ and $q \xrightarrow{w} q$. A path $\pi$ \emph{contains a cycle over $\Gamma$} if $q \xrightarrow{w} q$ is a subpath of $\pi$ for some state $q$ and $\alp(w) = \Gamma$. 

  The NFA $\A$ is {\em deterministic\/} (DFA) if $|Q_0|=1$ and $|\delta(q,a)|\leq 1$ for every $q$ in $Q$ and $a$ in $\Sigma$. Note that we allow some transitions to be undefined. To obtain a complete automaton, it is necessary to add a sink state, which we do not consider when counting the number of states. In other words, we will consider in the sequel only automata without useless states, that is, every state appears on an accepting path.

  For two words $v = a_1 a_2 \cdots a_n$ and $w \in \Sigma^* a_1 \Sigma^* a_2 \Sigma^* \cdots \Sigma^* a_n \Sigma^*$, we say that $v$ is a {\em subsequence\/} of $w$ or that $v$ can be {\em embedded\/} into $w$, denoted by $v \preq w$. A word $v\in\Sigma^*$ is a prefix of $w\in\Sigma^*$, denoted by $v \le w$, if $w=vu$, for some $u\in\Sigma^*$.

  We define towers of subsequences as a generalization of Stern's alternating towers between a language and its complement~\cite{Stern-tcs85}. For two languages $K$ and $L$, a sequence $(w_i)_{i=1}^{r}$ of words is a {\em tower of subsequences\/} between $K$ and $L$ if $w_1 \in K \cup L$ and, for all $i < r$,
  \begin{enumerate}
    \itemsep0pt
    \item $w_i \preq w_{i+1}$, 
    \item $w_i \in K$ implies $w_{i+1} \in L$, and 
    \item $w_i \in L$ implies $w_{i+1} \in K$.
  \end{enumerate}
    Similarly, a sequence $(w_i)_{i=1}^{r}$ of words is a {\em tower of prefixes\/} between $K$ and $L$ if $w_1 \in K \cup L$ and, for all $i < r$,
    $w_i \le w_{i+1}$,
    $w_i \in K$ implies $w_{i+1} \in L$, and
    $w_i \in L$ implies $w_{i+1} \in K$.
  
  The number of words in the sequence, $r$, is the {\em height\/} of the tower. If $r=\infty$, then we speak about an {\em infinite tower\/} between $K$ and $L$. The languages $K$ and $L$ are not necessarily disjoint. However, if there is a word $w \in K \cap L$, then there is a trivial infinite tower $w, w, w, \ldots$. If the languages are clear from the context, we usually omit them. 
  By a {\em tower between two automata}, we mean a tower between their languages.
  
  In what follows, if we talk about towers without a specification, we mean towers of subsequences. If we mean towers of prefixes, we always specify it explicitly.

\section{Upper bound on the height of towers of subsequences}
  Given two languages represented as NFAs, there is either an infinite tower between them, or the height of towers between the languages is bounded~\cite{icalp2013}. We now estimate that bound in terms of the number of states, $n$, and the size of the alphabet, $m$, of the NFAs.
  Stern's bound for minimal DFAs is $2^{m^2 n}$. Our new bound is $O(n^m)=O(2^{m\log n})$ and holds for NFAs. Therefore, if the alphabet is fixed, our bound is polynomial with respect to the number of states; otherwise, it is exponential in the size of the alphabet. 

  Before we state the main result of this subsection, we recall that the depth of an automaton is the number of states on the longest simple path, hence bounded by the number of states of the automaton.
  
  \begin{thm}\label{thm01}
    Let $\A_0$ and $\A_1$ be NFAs with $n_1$ and $n_2$ states, respectively, over an alphabet $\Sigma$. Assume that there is no infinite tower between the languages $L(\A_0)$ and $L(\A_1)$, and let $(w_i)_{i=1}^r$ be a tower between the languages such that $w_i\in L(\A_{i\bmod 2})$. Let $1 < n \le \max(n_1,n_2)$ denote the maximum of the depths of $\A_0$ and $\A_1$. Then $r \le \frac{n^{|\Sigma|+1}-1}{n-1}$. 
  \end{thm}
  \begin{proof}
    To prove the upper bound, we assign to each $w_i$ of the tower an integer $W_i$ in such a way that $0\leq W_1<W_2< \cdots < W_r < \frac{n^{|\Sigma|+1}-1}{n-1}$. To this aim, we define a factorization of $w_r$ using an accepting path of $w_r$ in $\A_{r\bmod 2}$. Then we inductively define factorizations of all $w_i$, $1\leq i < r$, depending on an accepting path of $w_i$ in $\A_{i\bmod 2}$ and on the factorization of $w_{i+1}$. 
    The value of $W_i$ is derived from these factorizations.

    We now define the concepts we need in the proof. We say that a sequence $(v_1,v_2,\dots, v_k)$ of nonempty words is a \emph{cyclic factorization} of $w=v_1v_2\cdots v_k$ with respect to some path $\pi$ from a state $q$ to a state $q'$ under $w$ in an automaton $\A$ if $\pi$ can be decomposed into $q_{i-1} \stackrel {v_i} \longrightarrow q_{i}\,$, $i=1,2,\dots k$, and either $v_i$ is a letter, or the path $q_{i-1} \stackrel {v_i} \longrightarrow q_{i}$ contains a cycle over $\alp(v_i)$. We call $v_i$ a \emph{letter factor} if it is a letter and $q_{i-1}\neq q_i$, and a \emph{cycle factor} otherwise. 
    Note that our cyclic factorization is closely related to the factorization by Almeida~\cite{Almeida-jpaa90}, see also Almeida~\cite[Theorem~8.1.11]{AlmeidaBook}.

    We now show that if $\pi$ is a path $q \xrightarrow{w} q'$ in some automaton $\A$ with depth $n$, then $w$ has a cyclic factorization $(v_1,v_2,\dots, v_k)$ with respect to $\pi$ that contains at most $n$ cycle factors and at most $n-1$ letter factors. Moreover, if $k>1$, then $\alp(v_i)$ is a strict subset of $\alp(w)$ for each cyclic factor $v_i$, $i=1,2,\dots,k$. We call such a cyclic factorization \emph{nice}. By convention, the empty sequence is a nice cyclic factorization of the empty word  with respect to the empty path. 

    Consider a path $\pi$ of the automaton $\A$
    from $q$ to $q'$ labeled by a word $w$.
    Let $q_0=q$ and define the factorization  $(v_1,v_2,\dots, v_k)$ inductively
    by the following greedy strategy. 
    Assume that we have defined the factors $v_1,  v_2\ldots, v_{i-1}$
    such that $w = v_1 \cdots v_{i-1} w'$
    and $q_0 \xrightarrow{ v_1 v_2\cdots v_{i-1}} q_{i-1}$.
    The factor $v_i$ is defined as the label
    of the longest possible initial segment $\pi_i$ 
    of the path $q_{i-1}\xrightarrow{w'} q'$ 
    such that either $\pi_i$ contains a cycle over $\alp(v_i)$
    or
    $\pi_i=q_{i-1},a,q_{i}$, where $v_i=a$ is a letter.
    Such a factorization is well defined,
    and it is a cyclic factorization of $w$.

    Let $p_i$, for $i=1,2,\dots,k$, be a state
    such that the path $q_{i-1} \stackrel {v_i} \longrightarrow q_{i}$
    contains a cycle $p_i\rightarrow p_i$ over the alphabet $\alp(v_i)$ if $v_i$ is a cycle factor,
    and $p_i=q_{i-1}$ if $v_i$ is a letter factor.
    If $p_i=p_j$ with $i<j$ such that $v_i$ and $v_j$ are cycle factors,
    then we have a contradiction with the maximality of $v_i$ since
    $q_{i-1} \xrightarrow{v_i v_{i+1}\cdots v_j}  q_{j}$
    contains a cycle $p_i\rightarrow p_i$ from $p_i$ to $p_i$
    over the alphabet $\alp(v_i v_{i+1}\cdots v_j)$.
    Therefore, the factorization contains at most $n$ cycle factors.

    Note that $v_i$ is a letter factor only if the state $p_i$,
    which is equal to $q_{i-1}$ in such a case,
    has no reappearance in the path $q_{i-1}\xrightarrow{v_i \cdots v_k} q'$.
    This implies that there are at most $n-1$ letter factors. 
    Finally, if $\alp(v_i)=\alp(w)$ for a cyclic factor $v_i$, then $v_i=v_1=w$
    follows from the maximality of $v_1$. Therefore $(v_1,v_2,\dots, v_k)$ is a nice cyclic factorization of $w$.

    We are now ready to execute the announced strategy. Let $\left(v_{r,1}, v_{r,2}, \dots, v_{r,k_r}\right)$ be a nice cyclic factorization of $w_r$ with respect to some accepting path in the automaton $\A_{r\bmod 2}$.
    Given a (not necessarily nice) cyclic factorization $(v_{i,1},v_{i,2},\dots ,v_{i,k_i})$ of $w_{i}$, $i=2,3,\dots,r$, the factorization $\left(v_{i-1,1},v_{i-1,2},\dots, v_{i-1,k_{i-1}}\right)$ of $w_{i-1}$ is defined as follows. Let 
    $
      w_{i-1}=v'_{i,1}v'_{i,2}\cdots v'_{i,k_{i}},
    $
    where $v'_{i,j}\preq v_{i,j}$, for each $j=1,2,\dots,k_{i}$. Such words (possibly empty) exist, since we have that $w_{i-1}\preq w_{i}$. Let $\pi_{i,j}$ be paths under $v'_{i,j}$ that together form an accepting path of $w_{i-1}$ in $\A_{i-1\bmod 2}$. Then 
		\[
      (v_{i-1,1},v_{i-1,2},\dots, v_{i-1,k_{i-1}}) = \prod_{j=1}^{k_i} \zav{v''_{i,j,1},v''_{i,j,2},\dots, v''_{i,j,m_{i,j}}},
		\] 
		where $\zav{v''_{i,j,1},v''_{i,j,2},\dots, v''_{i,j,m_{i,j}}}$ is a nice cyclic factorization of $v'_{i,j}$ with respect to $\pi_{i,j}$ and the product denotes the concatenation of sequences. Note that if $v_{i,j}$ is a letter factor of $w_{i}$ then either $m_{i,j}=0$ (if $v_{i,j}'$ is empty) or $m_{i,j}=1$ and $v_{i,j,1}''$ is a letter factor of $w_{i-1}$.

    To define $W_i$, let $g$ be the function
    $
      g(x)= n\frac{{n^x}-1}{n-1},
    $
    and let $f(v_{i,j})=1$ if $v_{i,j}$ is a letter factor, and 
    $f(v_i)=g\left(\left|\alp(v_i)\right|\right)$
    if $v_{i}$ is a cycle factor. 
    We now set  
    \[
      W_i=\sum_{j=1}^{k_i}f\left(v_{i,j}\right)\,.
    \]
    The key property of $g$ is that $g(x+1)=n\cdot g(x)+(n-1)+1$. (In fact, this equality and $g(0)=0$ defines $g$.)  This property implies, together with the definition of a nice cyclic factorization, that 
      \begin{align} \label{fg}
      \sum_{\ell=1}^{m_{i,j}}f\left(v''_{i,j,\ell}\right)\leq f\left(v_{i,j}\right),
      \end{align} 
      for all $i=2,3,\dots,r$ and $j=1,2,\dots,k_i$.   
      In particular, 
      \begin{align}
      W_r = \sum_{\ell=1}^{k_r} f\zav{v_{r,\ell}} \leq f\left(w_r\right)\leq g(|\Sigma|) < g(|\Sigma|)+1 = \frac{n^{|\Sigma|+1}-1}{n-1}\,. \label{gr}
      \end{align}
      Definitions of $g$ and of a nice cyclic factorization also imply that there is equality in \eqref{fg} if and only if $m_{i,j}=1$, $\alp\left(v_{i,j}\right)=\alp\left(v''_{i,j,1}\right)$, and both $v_{i,j}$ and $v''_{i,j,1}$ are either  letter factors, or cyclic factors. 
    We deduce that $W_{i-1}\leq W_i$, $i=2,3,\dots, r$, and, moreover, $W_{i-1}= W_i$ if and only if 
    \begin{itemize}
      \itemsep0pt
      \item $k_{i-1} = k_i$,
      \item $\alp(v_{i,j}) = \alp (v_{i-1,j})$ for all $j=1,2,\dots, k_i$, and,
      \item for all $j=1,2,\dots, k_i$, $v_{i,j}$ is a letter factor if and only if $v_{i-1,j}$ is a letter factor.  
    \end{itemize}
    We show that if these conditions are met, then there is an infinite tower between $\A_0$ and $\A_1$. Let $Z$ be the language of words $z_1z_2\cdots z_{k_i}$ such that $z_j=v_{i,j}$ if $v_{i,j}$ is a letter factor, and $z_j\in (\alp(v_{i,j}))^*$ if $v_{i,j}$ is a cycle factor. In particular, $w_i,w_{i-1}\in Z$. Since $w_i\in L(\A_{i\bmod 2})$ and $w_{i-1}\in L(\A_{i-1\bmod 2})$, the definition of a cycle factor implies that, for each $z\in Z$, there exist $z'\in L(\A_0)\cap Z$ such that $z\preq z'$ and $z''\in L(\A_1)\cap Z$ such that $z\preq z''$. The existence of an infinite tower follows. We have therefore proved that $W_{i-1}<W_{i}$, which together with \eqref{gr} completes the proof.
  \end{proof}

  The question is how good this bound is. We study this question next and show that it is tight if the alphabet is fixed. If the alphabet grows with the number of states of the automata, then we can construct a tower of exponential height with respect the the number of states of the automata (as well as with respect to the size of the alphabet). However, we do not know whether this bound is tight. We formulate this question as the following open problem.
  \begin{oprob}\label{op1}
    Let $\A_0$ and $\A_1$ be NFAs with $n_1$ and $n_2$ states, respectively, over an alphabet $\Sigma$ with $|\Sigma| \ge n_1+n_2$. Let $n$ b the maximum depth of $\A_0$ and $\A_1$. Assume that there is no infinite tower between the languages $L(\A_0)$ and $L(\A_1)$, and let $(w_i)_{i=1}^r$ be a tower between them. Is it true that $r \le \frac{n^{n_1+n_2+1}-1}{n-1}$ or even that $r\le 2^{n_1+n_2}$?
  \end{oprob}

\section{Lower bounds on the height of towers for NFAs}
  The upper bound of Theorem~\ref{thm01}, as well as its proof, indicate that the size of the alphabet is significant for the height of towers. This is confirmed by lower bounds considered in this section. We consider two cases in the following two subsections, namely (i) the size of the alphabet is fixed and (ii) the size of the alphabet may grow with the size of the automata. We show that the upper bound of Theorem~\ref{thm01} is asymptotically tight if the size of the alphabet is fixed. Then we show that the lower bound may be exponential with respect to the size of the automata if the alphabet may grow. In this case, the size of the alphabet is approximately the number of states of the automata. However, the precise upper bound for this case is left open, cf. Open Problem~\ref{op1}.

\subsection{Lower bounds on the height of towers for NFAs over a fixed alphabet}
  For a binary alphabet, the upper bound of Theorem~\ref{thm01} gives $n^2+n+1$ and it is known to be tight up to a linear factor~\cite{mfcs2014}. Namely, for every odd positive integer $n$, there are two binary NFAs with $n-1$ and $n$ states having a tower of height $n^2-4n+5$ and no infinite tower. We now improve this bound.

  \begin{thm}\label{thm02B}
    For every positive integer $d$ and every odd positive integer $e$, there exists a binary NFA with $d+1$ states and a binary DFA with $e+1$ states having a tower of height $d(e+1)+2$ and no infinite tower.
  \end{thm}
    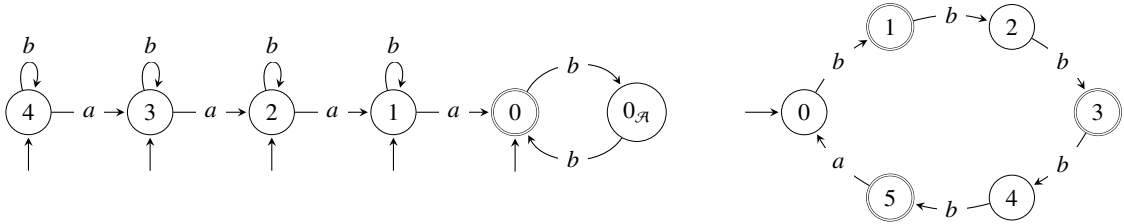
\begin{figure}[b]
      \centering
      \begin{tikzpicture}[baseline,->,>=stealth,shorten >=1pt,node distance=1.6cm,
        state/.style={circle,minimum size=1mm,very thin,draw=black,initial text=},
        every node/.style={fill=white,font=\small},
        bigloop/.style={shift={(0,0.01)},text width=1.6cm,align=center},
        bigloopd/.style={shift={(0,-0.01)},text width=1.6cm,align=center}]
        \node[state,initial below]            (1) {$4$};
        \node[state,initial below]            (2) [right of=1] {$3$};
        \node[state,initial below]            (3) [right of=2] {$2$};
        \node[state,initial below]            (4) [right of=3] {$1$};
        \node[state,accepting,initial below]  (5) [right of=4] {$0$};
        \node[state]                          (6) [right of=5] {$0_{\A}$};
        \path
          (1) edge node {$a$} (2)
          (2) edge node {$a$} (3)
          (3) edge node {$a$} (4)
          (4) edge node {$a$} (5)
          (1) edge[loop above] node[bigloop] {$b$} (1)
          (2) edge[loop above] node[bigloop] {$b$} (2)
          (3) edge[loop above] node[bigloop] {$b$} (3)
          (4) edge[loop above] node[bigloop] {$b$} (4)
          (5) edge[bend left=60] node {$b$} (6)
          (6) edge[bend left=60] node {$b$} (5);
      \end{tikzpicture}
      \quad\quad
      \begin{tikzpicture}[baseline,->,>=stealth,shorten >=1pt,node distance=1.6cm,
        state/.style={circle,minimum size=1mm,very thin,draw=black,initial text=},
        every node/.style={fill=white,font=\small},
        bigloop/.style={shift={(0,0.01)},text width=1.6cm,align=center}]
        \node[state,initial]    (1) {$0$};
        \node[state,accepting]  (2) [above right of=1]{$1$};
        \node[state]            (3) [right of=2] {$2$};
        \node[state, accepting] (4) [below right of=3] {$3$};
        \node[state]            (5) [below left of=4] {$4$};
        \node[state, accepting] (6) [left of=5] {$5$};
        
        \path
          (1) edge[bend left=15] node {$b$} (2)
          (2) edge[bend left=15] node {$b$} (3)
          (3) edge[bend left=15] node {$b$} (4)
          (4) edge[bend left=15] node {$b$} (5)
          (5) edge[bend left=15] node {$b$} (6)
          (6) edge[bend left=15] node {$a$} (1);
       \end{tikzpicture}
      \caption{Automata $\A_d$ and $\B_e$ of Theorem \ref{thm02B} for $d=e=5$}
      \label{exbinB}
    \end{figure}
   \begin{proof}
    We define the automata $\A_d$ and $\B_e$ with $d+1$ and $e+1$ states, respectively, as depicted in Figure~\ref{exbinB}. 
    The NFA $\A_d=(\{0,1,\ldots,d-1\}\cup\{0_{\A}\},\{a,b\},\delta_d,\{0,1,\ldots,d-1\},\{0\})$ consists of an $a$-path through the states $d-1, \allowbreak \ldots, 0$, of self-loops under $b$ in all states $1, \ldots, d-1$, and of a $b$-cycle from $0$ to $0_{\A}$ and back to $0$.
    The DFA $\B_e=(\{0,1,\ldots,e\},\{a,b\},\delta_e,0,\{1,3,\dots,e\})$ consists of a $b$-path through the states $0,1,\ldots,e$ and of an $a$-transition from state $e$ to state $0$. All odd states are accepting.

    Consider the word $w=(b^ea)^{d-1} b^{e+1}$. Note that $\A_d$ accepts all prefixes of $w$ ending with an even number of $b$'s, including those ending with $a$, and the empty prefix. On the other hand, the automaton $\B_e$ accepts all prefixes of $w$ ending with an odd number of $b$'s. Moreover, the automaton $\B_e$ accepts the word $(b^ea)^{d-1} b^e a b$. Hence the sequence $(w_i)_{i=1}^{|w|+2}$, where, for $i=1,2,\dots,|w|+1$, $w_i$ is the prefix of $w$ of length $i-1$, and $w_{|w|+2} = (b^ea)^{d-1}b^eab$ is a tower between $\A_d$ and $\B_e$ of height $|w|+2 = (e+1)(d-1) + e + 1 + 2 = d(e+1)+2$.

    We show that there is no higher tower between the languages, in particular, there is no infinite tower. Notice that any word in $L(\B_e)$ is a prefix of $(b^e a)^*$. As the languages are disjoint (they require a different parity of the $b$-tail), any tower $(w_i)_{i=1}^{r}$ is strictly increasing with respect to $\preq$ and thus $|w_i|\geq i-1$. Thus if the height of $(w_i)_{i=1}^{r}$ is larger than $d(e+1)+2$ the word $w_{d(e+1)+1}$ or $w_{d(e+1)+2}$ is in $L(\B_e)$ and therefore contains at least $d$ occurrences of letter $a$. However, no such word can be embedded into a word of $L(\A_d)$, since each word of $L(\A_d)$ contains at most $d-1$ occurrences of letter $a$.
  \end{proof}

  As a consequence of Theorem~\ref{thm02B}, we obtain the following lower bound on the height of binary towers.
  
  \begin{cor}\label{cor:quadratic}
    For every even positive integer $n$, there exists a binary NFA with $n$ states and a binary DFA with $n$ states having a tower of height $n^2-n+2$ and no infinite tower.
  \end{cor}
  \begin{proof}
    Set $d=e=n-1$ in Theorem~\ref{thm02B}.
  \end{proof}

  For a four-letter alphabet and for every $n\ge 1$, there are two NFAs with at most $n$ states having a tower of height $\Omega(n^3)$ and no infinite tower~\cite[Theorem 3]{mfcs2014}. 
  We now improve this bound by generalizing Theorem~\ref{thm02B}. 

  \begin{thm}\label{thm:lower_bound}
    For all integers $m_A\ge 1$, $m_B\ge 0$, and $d_1,d_2,\dots,d_{m_A},e_0,e_1,e_2,\dots, e_{m_B} \ge 1$, where $e_0$ is odd, there exist two NFAs with $\sum_{i=1}^{m_A} d_i + 2$ and $\sum_{i=0}^{m_B} e_i + 1$ states over an alphabet of cardinality $m_A + m_B + 1$ having a tower of height $\prod_{i=1}^{m_A}(d_i+1)\prod_{i=0}^{m_B}(e_i+1) + 2$ and no infinite tower.
  \end{thm}
  \begin{proof}
    Let $d_0=1$, $\vek{d}=(d_1,\dots,d_{m_A})$, and $\vek{e}=(e_1,\ldots,e_{m_B})$. 
    For $k\ge 0$, we define the alphabets $\Sigma_k=\{b,a_1,a_2,\ldots,a_k\}$ and $\Gamma_k=\{c_1,c_2,\dots,c_k\}$. 
    We now define two NFAs $\A_{m_A,m_B,\vek{d}}$ and $\B_{m_A,m_B,\vek{e}}$ over $\Sigma_{m_A}\cup\Gamma_{m_B}$ as follows.

    The set of states of $\A_{m_A,m_B,\vek d}$ is $Q_{m_A,\vek d}=\{(k,j) \mid k=0,1,2,\dots,m_A; \allowbreak \ j=0,1,2,\dots,d_k-1\}\cup\{0_{\A}\}$. All states except for $0_{\A}$ are initial, and state $(0,0)$ is the unique accepting state. 
    The transition function of $\A_{m_A,m_B,\vek d}$ consists of
    \begin{itemize}
      \itemsep0pt
      \item a self-loop under $\Sigma_{k-1}$ for all states $(k,j)$ with $k>0$,
      \item an $a_i$-transition 
        from each $(i,j)$ to $(i,j-1)$, for $i,j>0$, and 
        from each $(i,0)$ to states $(\ell,j)$, for $0\leq \ell < i$ and $j=0,1,2,\dots,d_{\ell}-1$, 
      \item $b$-transitions from $(0,0)$ to $0_{\A}$ and back, and
      \item transitions under $\Gamma_{m_B}$ from state $0_{\A}$ to each state of $Q_{m_A,\vek d}$ different from state $0_{\A}$.
    \end{itemize}
    The NFA $\A_{m_A,m_B,\vek d}$ with $m_A=3$ and $\vek d=(3,1,2)$ is shown in Figure~\ref{fig:A}.
    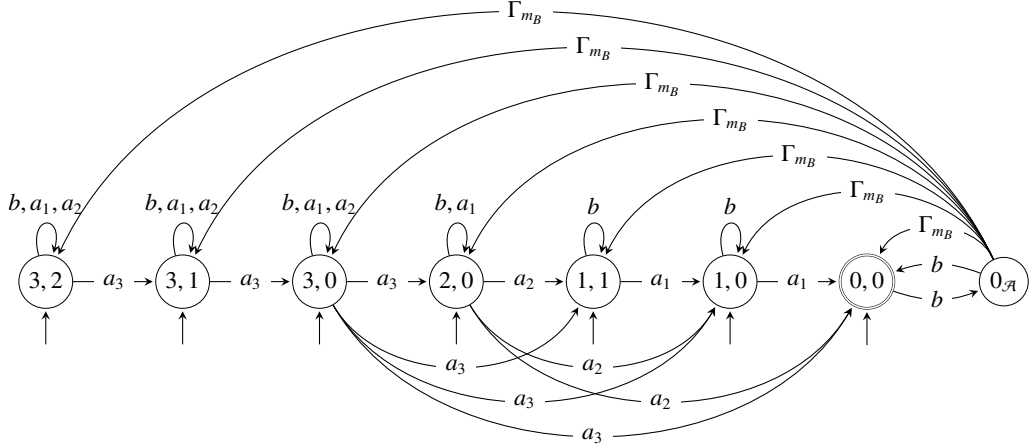
\begin{figure}[!ht]
      \centering
      \begin{tikzpicture}[baseline,->,>=stealth,shorten >=1pt,node distance=1.8cm,
        state/.style={circle,minimum size=1mm,very thin,draw=black,inner sep=2pt,initial text=},
        every node/.style={font=\small},
        bigloop/.style={shift={(0,0.01)},text width=1.6cm,align=center}]
        \node[state,initial below,accepting] (00) {$0,0$};
        \node[state,initial below]           (10) [left of=00] {$1,0$};
        \node[state,initial below]           (11) [left of=10] {$1,1$};
        \node[state,initial below]           (20) [left of=11] {$2,0$};
        \node[state,initial below]           (30) [left of=20] {$3,0$};
        \node[state,initial below]           (31) [left of=30] {$3,1$};
        \node[state,initial below]           (32) [left of=31] {$3,2$};
        \node[state]                         (0A) [right of=00] {$0_{\A}$};
          \foreach \from/\to/\pis in {32/31/3,31/30/3,30/20/3,20/11/2,11/10/1,10/00/1}
          \path (\from) edge node[fill=white] {$a_\pis$} (\to);
          \path 
          (00) edge[bend right=20] node[fill=white] {$b$} (0A)
          (0A) edge[bend right=20] node[fill=white] {$b$} (00);
          \foreach \from/\pis in {32/{$b,a_1,a_2$},31/{$b,a_1,a_2$},30/{$b,a_1,a_2$},20/{$b,a_1$},11/{$b$},10/{$b$}}
          \path (\from) edge[loop above] node[bigloop] {\pis} (\from);
          \foreach \from/\to/\pis in {30/11/3,30/10/3,30/00/3,20/10/2,20/00/2}
            \path (\from) edge[bend right=60] node[fill=white] {$a_\pis$} (\to);
          \foreach \to in {00,10,32,31,30,20,11}
          \path (0A) edge[bend right=65] node[fill=white] {$\Gamma_{m_B}$} (\to); 
      \end{tikzpicture}
      \caption{Automaton $\A_{3,m_B,(3,1,2)}$ of Theorem~\ref{thm:lower_bound}}
      \label{fig:A}
    \end{figure}

    The NFA $\B_{m_A,m_B,\vek e}$ has the state set $Q_{m_B,\vek e}=\{(k,j) \mid k=0,1,2,\dots,m_B; \allowbreak \ j=0,1,2,\dots,e_i-1\}\cup\{(0,i) \mid 0 \le i \le e_0\}$. All states are initial, except for states $(0,i)$ with $i \neq 0$. Accepting states are the states $(0,i)$ with $i$ odd. The transition function of $\B_{m_A,m_B,\vek e}$ consists of
    \begin{itemize}
      \itemsep0pt
      \item self-loops under $\Sigma_{m_A}\cup \Gamma_{k-1}$ for all states $(k,j)$ with $k>0$,
      \item self-loops under $\Sigma_{m_A}\setminus\{b\}$ in state $(0,0)$,
      \item a $c_i$-transition 
        from each $(i,j)$ to $(i,j-1)$, for $i,j>0$, and
        from each $(i,0)$ to states $(\ell,j)$, for $1 \le \ell < i$ and $j=0,1,2,\dots,e_{\ell}-1$, and to state $(0,0)$,
      \item  a $b$-transition from each $(0,j)$ to $(0,j+1)$, for $0\leq j < e_0$, and
      \item edges under $\Sigma_{m_A}\setminus\{b\}$ from state $(0,e_0)$ to state $(0,0)$.
    \end{itemize}
    The NFA $\B_{m_A,m_B,\vek e}$ with $m_B=3$ and $\vek e=(2,2,2)$ is shown in Figure~\ref{fig:B}.
    \begin{figure}
      \centering
      \begin{tikzpicture}[baseline,->,>=stealth,shorten >=1pt,node distance=1.8cm,
        state/.style={circle,minimum size=1mm,very thin,draw=black,inner sep=2pt,initial text=},
        every node/.style={font=\small},
        bigloop/.style={shift={(0,0.01)},text width=1.6cm,align=center}]
        \node[state,initial below]  (03) {$0,0$};
        \node[state,initial below]  (10) [left of=03] {$1,0$};
        \node[state,initial below]  (11) [left of=10] {$1,1$};
        \node[state,initial below]  (20) [left of=11] {$2,0$};
        \node[state,initial below]  (21) [left of=20] {$2,1$};
        \node[state,initial below]  (30) [left of=21] {$3,0$};
        \node[state,initial below]  (31) [left of=30] {$3,1$};
        \node[state,accepting]      (02) [above right of=03,node distance=2.3cm] {$0,1$};
        \node[state]                (01) [below right of=02,node distance=2.3cm] {$0,2$};
        \node[state,accepting]      (00) [below right of=03,node distance=2.3cm] {$0,3$};
        
        \foreach \from/\to/\pis in {31/30/3,30/21/3,21/20/2,20/11/2,11/10/1,10/03/1}
          \path (\from) edge node[fill=white] {$c_\pis$} (\to);
        
        \foreach \from/\to in{03/02,02/01,01/00}
          \path (\from) edge[bend left=20] node[fill=white] {$b$} (\to);
        
        \foreach \from/\pis in {31/2,30/2,21/1,20/1}
          \path (\from) edge[loop above] node[bigloop] {$\Sigma_{m_A}\cup \Gamma_\pis$} (\from);
        
        \foreach \from in {11,10}
          \path (\from) edge[loop above] node[bigloop] {$\Sigma_{m_A}$} (\from);
        
        \path (03) edge[loop above] node[bigloop] {$\Sigma_{m_A}\setminus\{b\}$~~~~~~} (03);
        \path (00) edge[bend left] node[fill=white] {~~~~~$\Sigma_{m_A}\setminus\{b\}$} (03);
        
        \foreach \from/\to/\pis in {30/20/3,30/11/3,30/10/3,30/03/3,20/10/2,20/03/2}
          \path (\from) edge[bend right=60] node[fill=white] {$c_\pis$} (\to);
      \end{tikzpicture}
      \caption{Automaton $\B_{m_A,3,(2,2,2)}$ of Theorem~\ref{thm:lower_bound} with $e_0=3$}
      \label{fig:B}
    \end{figure}
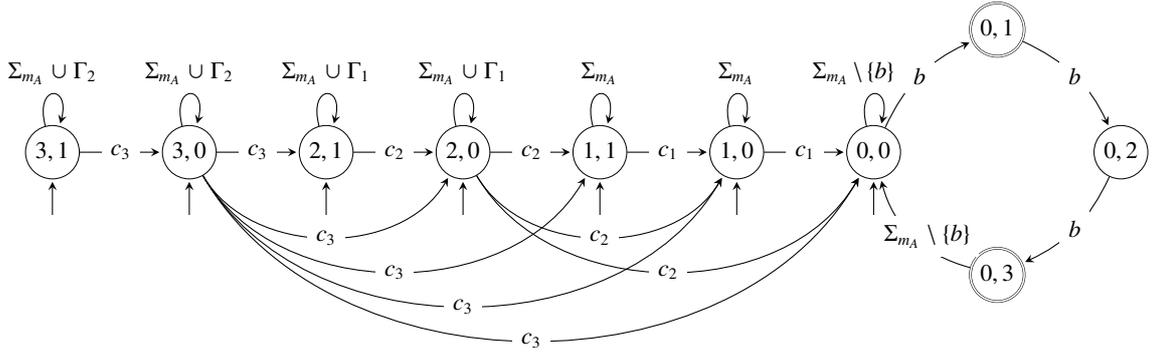

    We now define a word $u_{m_A+m_B}$ such that the sequence of prefixes of $u_{m_A+m_B}$ forms a tower. To do this, we proceed in two steps. First, we define a word $\uu_{m_A}$ inductively by $u_0 = b^{e_0}$ and $\uu_{k} = (\uu_{k-1} a_{k})^{d_k} \uu_{k-1}$, for $0 < k \le m_A$. Note that $\uu_k \in \Sigma_k^*$ and verify that $|\uu_k| = (e_0+1) \prod_{i=1}^k(d_i+1) - 1$ and that $\uu_k$ ends with an odd number of $b$'s, namely with $b^{e_0}$. 
    
    Let us notice that the sequence of prefixes of $u_{m_A+m_B}$ is such that if a prefix end by an even (resp. odd) number of $b$'s, then the next prefix ends with an odd (resp. even) number of $b$'s. It can be straightforwardly verified that every prefix of $\uu_{m_A}$ ending with an odd number of $b$'s is accepted by $\B_{m_A,m_B,\vek e}$ from state $(0,0)$. 
    
    We now show by induction on $k$ that every prefix of $\uu_k$ ending with an even number of $b$'s (including no $b$) is accepted by $\A_{m_A,m_B,\vek d}$ from a state smaller than $(k+1,0)$. 
		(By ``smaller'' we mean lexicographically smaller or, equivalently, closer to the state $(0,0)$.)
		Any prefix $b^{2i}$ of $u_0$ is accepted from state $(0,0)$. Consider a prefix $v$ of $\uu_k$ ending with an even number of $b$'s. Let $v = (u_{k-1} a_k)^\ell v'$, where $0 \leq \ell \leq d_k$ and $v'$ is a prefix of $u_{k-1}$ ending with an even number of $b$'s. By the induction assumption, $v'$ is accepted from a state $s$ smaller than $(k,0)$. Therefore, $v$ is accepted by $\A_{m_A,m_B,\vek d}$ by the path
    \[
      \underbrace{
      (k,\ell-1)  \xrightarrow{~\uu_{k-1}~} (k,\ell-1) 
                \xrightarrow{~a_{k}~} (k,\ell-2) 
                \quad \cdots \quad (k,1) 
                \xrightarrow{~\uu_{k-1}~} (k,1) 
                \xrightarrow{~a_{k}~} (k,0) 
                \xrightarrow{~\uu_{k-1}~} (k,0) 
                \xrightarrow{~a_{k}~} s}_{\ell-\text{times}}
                \xrightarrow{~v'~} (0,0)\,.
    \]
    
    We now use the word $\uu_{m_A}$ to define, for $0<k\leq m_B$, the words $u_{m_A+k}=(u_{m_A+k-1}c_{k})^{e_{k}}u_{m_A+k-1}$. Then we have that $u_{m_A+k}\in (\Sigma_{m_A}\cup \Gamma_{k})^*$, that $|u_{m_A+k}| = \prod_{i=1}^{m_A}(d_i+1)\prod_{i=0}^k(e_i+1)-1$, and that $u_{m_A+k}$ ends with an odd number of $b$'s, namely with $b^{e_0}$.
    
    Similarly as above, we show that the prefixes of $u_{m_A+m_B} b$ form a tower of height $|u_{m_A+m_B}b|+1$ (the one additional element is the empty word). Specifically, we prove by induction on $k$ that every prefix $v$ of $u_{m_A+k}$ ending with an even number of $b$'s is accepted by $\A_{m_A,m_B,\vek d}$ from a state smaller than $(k+1,0)$, and that every prefix of $u_{m_A+k}$ ending with an odd number of $b$'s is accepted by $\B_{m_A,m_B,\vek e}$ from a state smaller than $(k+1,0)$.
    
    This is true for $k=0$ as shown above. Let $k\geq 1$ and consider the word $u_{m_A+k}$. Let $v$ be a prefix of $u_{m_A+k}$ of the form $(u_{m_A+k-1}c_{k})^\ell v'$, where $0\leq \ell\leq e_k$ and $v'$ is a prefix of $u_{m_A+k-1}$. Recall that $u_{m_A+k-1}$ ends with $b^{e_0}$, which is an odd number of $b$'s. Therefore, $u_{m_A+k-1}=wb$ and, by the induction hypothesis, $w$ is accepted by $\A_{m_A,m_B,\vek d}$ from a state $t$.
    
    If $v'$ ends with an even number of $b$'s, then $\A_{m_A,m_B,\vek d}$ accepts $v'$ from a state $s$ by the induction hypothesis. Then the whole prefix $v$ is accepted in $\A_{m_A,m_B,\vek d}$ by the path 
    \[
      \underbrace{
        t \overbrace{
          \xrightarrow{~w~} (0,0) 
          \xrightarrow{~b~} 
        }^{\uu_{m_A+k-1}} 0_{\A}
        \xrightarrow{~c_k~} t  
        \xrightarrow{~\uu_{m_A+k-1}~}  
        0_{\A}\ \dots\ t
        \xrightarrow{~\uu_{m_A+k-1}~} 0_{\A}
        \xrightarrow{~c_k~} s
      }_{\ell-\text{times}}
        \xrightarrow{~v'~} (0,0)\,.
    \]
    
    If $v'$ ends with an odd number of $b$'s, it is accepted by $\B_{m_A,m_B,\vek e}$ from a state $s'$, by the inductive hypothesis. Then the whole prefix $v$ is accepted in $\B_{m_A,m_B,\vek e}$ by the path
    \[
      \underbrace{
      (k,\ell-1)  \xrightarrow{~\uu_{m_A+k-1}~} (k,\ell-1) 
        \xrightarrow{~c_{k}~} (k,\ell-2) 
        \quad \cdots 
        \cdots \quad (k,1) 
        \xrightarrow{~\uu_{m_A+k-1}~} (k,1) 
        \xrightarrow{~c_{k}~} (k,0) 
        \xrightarrow{~\uu_{m_A+k-1}~} (k,0) 
        \xrightarrow{~c_{k}~} s' }_{\ell-\text{times}}
        \xrightarrow{~v'~} (0,i)\,,
    \]
    where state $(0,i)$ is accepting in $\B_{m_A,m_B,\vek e}$, that is, $i$ is odd.
    
    Together, the height of the tower formed by the prefixes is 
    $|u_{m_A+m_B}b|+1 
      = \prod_{i=1}^{m_A}(d_i+1)\prod_{i=0}^{m_B}(e_i+1) + 1.
    $
    Notice that $u_{m_A+m_B} b$ ends with $e_0+1$ letters $b$, hence it is accepted by $\A_{m_A,m_B,\vek d}$. Then $u_{m_A+m_B} b \preq u_{m_A+m_B} a b$ and $u_{m_A+m_B} a b$ is accepted by $\B_{m_A,m_B,\vek d}$, which extends the tower by one additional element and proves the claimed height of the tower.
    
    It remains to show that there is no infinite tower. We show it in two steps. We first show that there is no infinite tower over the alphabet $\Sigma_{m_A}$ and then that there is no infinite tower over the alphabet $\Sigma_{m_A}\cup\Gamma_{m_B}$. Suppose the contrary, and first let $k \geq 0$ be the smallest integer such that there is an infinite tower over $\Sigma_k$. Since $\eps,b,b^2,\dots,b^{e_0}$ is the highest tower over $\Sigma_0$, we have that $k\ge 1$. By the  induction hypothesis, there is no infinite tower over $\Sigma_{k-1}$, therefore we may consider an infinite tower where each word contains the letter $a_k$. 
		The crucial property of $\A_{m_A,m_B,\vek d}$ is that every word from $L(\A_{m_A,m_B,\vek d}) \cap \Sigma_k^*$ contains at most $d_k$ occurrences of letter $a_k$. Without loss of generality, we can therefore consider an infinite tower $(w_i)_{i=1}^{\infty}$ such that every $w_i$ contains the same (nonzero) number of occurrences of $a_k$. Let $w_i = w_i'' a_k w_i'$, where $w_i' \in \Sigma_{k-1}^*$. Since all transitions under $a_k$ end in an initial state in both automata, the word $w_i'$ is accepted by $\A_{m_A,m_B,\vek d}$ (by $\B_{m_A,m_B,\vek e}$ resp.) if $w_i$ is accepted by $\A_{m_A,m_B,\vek d}$ (by $\B_{m_A,m_B,\vek e}$ resp.). Then $(w_i')_{i=1}^\infty$ is an infinite tower over $\Sigma_{k-1}$; a contradiction.
    Thus, let $k\ge 1$ be the smallest integer such that there is an infinite tower over the alphabet $\Sigma_{m_A}\cup \Gamma_k$. Similarly as above, since the number of occurrences of $c_k$ in words over $\Sigma_{m_A}\cup \Gamma_k$ accepted by the automaton $\B_{m_A,m_B,\vek d}$ is restricted to $e_k$, we can choose an infinite tower $(w_i)_{i=1}^\infty$ in which all words have the same (nonzero) number of occurrences of the letter $c_k$. Let $w_i = w_i'' c_k w_i'$ with $w_i' \in (\Sigma_{m_A} \cup \Gamma_{k-1})^*$. A direct inspection of the automata yields that $w_i'$ is accepted by $\A_{m_A,m_B,\vek d}$ (by $\B_{m_A,m_B,\vek e}$ resp.) if $w_i$ is accepted by $\A_{m_A,m_B,\vek d}$ (by $\B_{m_A,m_B,\vek e}$ resp.). Then $(w_i')_{i=1}^\infty$ is an infinite tower yielding a contradiction.				
  \end{proof}

  If, in the previous theorem, $m_B=0$, the automaton $\B_{m_A,m_B,\vek e}$ is deterministic and we have the following corollary.
  \begin{cor}\label{thm:exp:cor}
    For all integers $m,d_1,d_2,\dots, d_m \geq 1$ and every odd positive integer $e$, there exist an NFA with $\sum_{i=1}^m d_i + 2$ states and a DFA with $e + 1$ states over an alphabet of cardinality $m + 1$ having a tower of height $(e+1)\prod_{i=1}^m (d_i+1) + 2$ and no infinite tower.
  \end{cor}
  \begin{proof}
    Let $m_A=m$, $m_B=0$ and $e=e_0$ in Theorem~\ref{thm:lower_bound}. Then there exist two NFAs with $\sum_{i=1}^{m} d_i + 2$ and $e + 1$ states over an alphabet of cardinality $m_A + m_B + 1$ having a tower of height $(e+1)\prod_{i=1}^{m_A}(d_i+1) + 2$ and no infinite tower. Moreover, notice that the automaton $\B_{m_A,0,\vek e}$ of Theorem~\ref{thm:lower_bound} is deterministic.
  \end{proof}
  
  Furthermore, the following corollary shows that the upper bound in Theorem~\ref{thm01} is tight if the alphabet is fixed even if one of the automata is deterministic.

  \begin{cor}\label{corThm3}
    Let $k \ge 2$ be a constant. Then the maximum height of a tower between an NFA with at most $n$ states and a DFA with at most $n$ states over an alphabet of cardinality $k$ having no infinite tower is in $\Omega\zav{n^k}$.
  \end{cor}

  \begin{proof}
	Let $m=k-1$, and let $\ell=\left\lfloor \frac{n-2}{m} \right\rfloor$. For sufficiently large $n$, there are integers $\ell \leq d_i \leq \ell+1$, $i=1,2,\dots,m$, and an odd integer  $n-2 \leq e\leq n-1$ such that Corollary~\ref{thm:exp:cor} yields an NFA with $n$ states and a DFA with at most $n$ states over an alphabet of cardinality $k=m+1$ having a tower of height at least $(n-1)(\ell+1)^{m} \in \Omega\zav{n^{k}}$.
   \end{proof}

\subsection{Lower bounds on the height of towers for NFAs over a growing alphabet}
  If the alphabet may grow, we immediately obtain the following result showing that the tower may be exponential with respect to the size of the alphabet which is of the size of the number of states of the automata.
  \begin{cor}\label{corThm7_2}
    For every integer $n\ge 3$, there exist two NFAs with $n$ states over an alphabet of cardinality $2n-3$ having a tower of height $2^{2n-3}+2$.
  \end{cor}
  \begin{proof}
    Let $m\ge 1$. By Theorem~\ref{thm:lower_bound} for $m_A=m_B=m$ and $d_i=e_i=e_0=1$, for $i=1,2,\ldots,m$, there are two NFAs with $n=m+2$ states over an alphabet of cardinality $2m+1$ having a tower of height $2^{2m+1} + 2 = 2^{2n-3} + 2$.
  \end{proof}

  We further improve the lower bound by the following result. Its proof, which we give in full length, is a slight modification of the proof of Theorem \ref{thm:lower_bound}. 
  \begin{thm}\label{thm:2exp:improved}
    For all integers $m_A\ge 1$, $m_B\ge 0$, and $d_1,d_2,\dots, d_{m_A}, e_1, e_2, \dots, e_{m_B} \ge 1$, there exist two NFAs with $\sum_{i=1}^{m_A} d_i + 1$ and $\sum_{i=1}^{m_B} e_i + 2$ states over an alphabet of cardinality $m_A + m_B + 1$ having a tower of height $\prod_{i=1}^{m_B}(e_i+1)\cdot \zav{2 \prod_{i=1}^{m_A}(d_i+1)-1}+1$ and no infinite tower.
  \end{thm}
  \begin{proof}
    Let $d_0=1$, $e_0=1$, $\vek{d}=(d_1,\dots,d_{m_A})$, and $\vek{e}=(e_1,\ldots,e_{m_B})$.
    For $k\ge 0$, let $\Sigma_k=\{b,a_1,a_2,\ldots,a_k\}$ and $\Gamma_k=\{c_1,c_2,\dots,c_k\}$ be alphabets. We define two NFAs $\A_{m_A,m_B,\vek{d}}$ and $\B_{m_A,m_B,\vek{e}}$ over $\Sigma_{m_A}\cup\Gamma_{m_B}$ as follows.

    The set of states of the NFA $\A_{m_A,m_B,\vek d}$ is $Q_{m_A,\vek d}=\{(k,j) \mid k=0,1,2,\dots,m_A; \allowbreak \ j=0,1,2,\dots,d_k-1\}$. All states are initial, and state $(0,0)$ is the unique accepting state. The transition function of $\A_{m_A,m_B,\vek d}$ consists of
    \begin{itemize}
      \itemsep0pt
      \item  a self-loop under $\Sigma_{k-1}$ for all states $(k,j)$ with $k>0$,
      \item  an $a_i$-transition 
        from each $(i,j)$ to $(i,j-1)$, for $i,j>0$, and 
        from each $(i,0)$ to states $(\ell,j)$, for $0\leq \ell < i$ and $j=0,1,2,\dots,d_{\ell}-1$, and
      \item transitions under $\Gamma_{m_B}$ from state $(0,0)$ to each state of $Q_{m_A,\vek d}\setminus\{(0,0)\}$.
    \end{itemize}
    The NFA $\A_{m_A,m_B,\vek d}$ with $m_A=3$ and $\vek d=(3,1,2)$ is shown in Figure~\ref{fig:A:improved}.
    \begin{figure}[ht]
      \centering
      \begin{tikzpicture}[baseline,->,>=stealth,shorten >=1pt,node distance=1.8cm,
        state/.style={circle,minimum size=1mm,very thin,draw=black,inner sep=2pt,initial text=},
        every node/.style={font=\small},
        bigloop/.style={shift={(0,0.01)},text width=1.6cm,align=center}]
        \node[state,initial below,accepting] (00) {$0,0$};
        \node[state,initial below]           (10) [left of=00,node distance=2.5cm] {$1,0$};
        \node[state,initial below]           (11) [left of=10] {$1,1$};
        \node[state,initial below]           (20) [left of=11] {$2,0$};
        \node[state,initial below]           (30) [left of=20] {$3,0$};
        \node[state,initial below]           (31) [left of=30] {$3,1$};
        \node[state,initial below]           (32) [left of=31] {$3,2$};
          \foreach \from/\to/\pis in {32/31/3,31/30/3,30/20/3,20/11/2,11/10/1,10/00/1}
          \path (\from) edge node[fill=white] {$a_\pis$} (\to);
          \foreach \from/\pis in {32/{$b,a_1,a_2$},31/{$b,a_1,a_2$},30/{$b,a_1,a_2$},20/{$b,a_1$},11/{$b$},10/{$b$}}
          \path (\from) edge[loop above] node[bigloop] {\pis} (\from);
          \foreach \from/\to/\pis in {30/11/3,30/10/3,30/00/3,20/10/2,20/00/2}
            \path (\from) edge[bend right=60] node[fill=white] {$a_\pis$} (\to);
          \foreach \to in {10,32,31,30,20,11}
          \path (00) edge[bend right=65] node[fill=white] {$\Gamma_{m_B}$} (\to); 
      \end{tikzpicture}
      \caption{Automaton $\A_{3,m_B,(3,1,2)}'$ of Theorem~\ref{thm:2exp:improved}}
      \label{fig:A:improved}
    \end{figure}
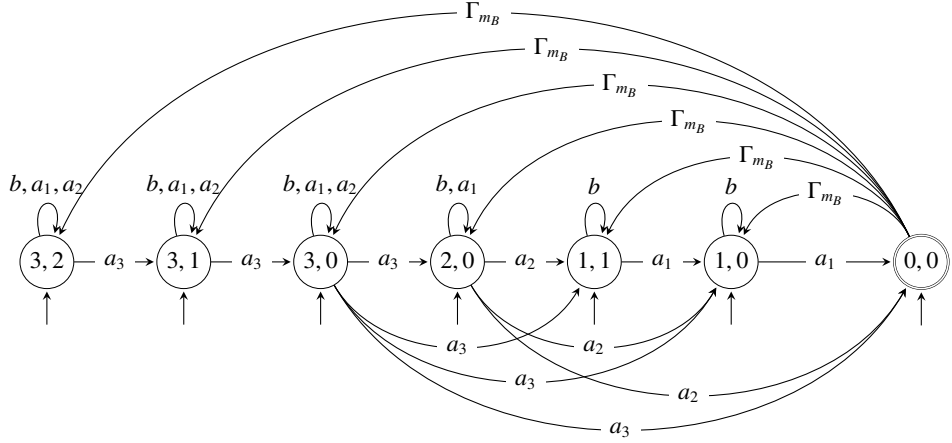

    The NFA $\B_{m_A,m_B,\vek e}$ has the state set $Q_{m_B,\vek e}=\{(k,j) \mid k=0,1,2,\dots,m_B; \allowbreak \ j=0,1,2,\dots,e_i-1\}\cup\{(0,1)\}$. All states are initial, except for state $(0,1)$, which, in turn, is the unique accepting state. The transitions of $\B_{m_A,m_B,\vek e}$ consist of
    \begin{itemize}
      \itemsep0pt
      \item self-loops under $\Sigma_{m_A}\cup \Gamma_{k-1}$ for all states $(k,j)$ with $k> 0$, and of self-loops under $\Sigma_{m_A}$ in state $(0,0)$,
      \item a $c_i$-transition 
        from each $(i,j)$ to $(i,j-1)$, for $i,j>0$, and
        from each $(i,0)$ to states $(\ell,j)$, for $0 \le \ell < i$ and $j=0,1,2,\dots,e_{\ell}-1$, and to state $(0,1)$, and
      \item  transitions under $\Gamma_{m_B}\cup\{b\}$ from $(0,0)$ to $(0,1)$.
    \end{itemize}
    The NFA $\B_{m_A,m_B,\vek e}$ with $m_B=3$ and $\vek e=(2,2,2)$ is shown in Figure~\ref{fig:B:improved}.
    \begin{figure}[ht]
      \centering
      \begin{tikzpicture}[baseline,->,>=stealth,shorten >=1pt,node distance=1.8cm,
        state/.style={circle,minimum size=1mm,very thin,draw=black,inner sep=2pt,initial text=},
        every node/.style={font=\small},
        bigloop/.style={shift={(0,0.01)},text width=1.6cm,align=center}]
        \node[state,initial below]  (03) {$0,0$};
        \node[state,initial below]  (10) [left of=03] {$1,0$};
        \node[state,initial below]  (11) [left of=10] {$1,1$};
        \node[state,initial below]  (20) [left of=11] {$2,0$};
        \node[state,initial below]  (21) [left of=20] {$2,1$};
        \node[state,initial below]  (30) [left of=21] {$3,0$};
        \node[state,initial below]  (31) [left of=30] {$3,1$};
        \node[state,accepting]      (02) [right of=03,node distance=3cm] {$0,1$};
        
        \foreach \from/\to/\pis in {31/30/3,30/21/3,21/20/2,20/11/2,11/10/1,10/03/1}
          \path (\from) edge node[fill=white] {$c_\pis$} (\to);
        
        \path (03) edge node[fill=white] {$\Gamma_{m_b}\cup\{b\}$} (02);
        
        \foreach \from/\pis in {31/2,30/2,21/1,20/1}
          \path (\from) edge[loop above] node[bigloop] {$\Sigma_{m_A}\cup \Gamma_\pis$} (\from);
        
        \foreach \from in {11,10,03}
          \path (\from) edge[loop above] node[bigloop] {$\Sigma_{m_A}$} (\from);
        
        \foreach \from/\to/\pis in {30/20/3,30/11/3,30/10/3,30/03/3,20/10/2,20/03/2,
          30/02/3,20/02/2,10/02/1}
          \path (\from) edge[bend right=60] node[fill=white] {$c_\pis$} (\to);
      \end{tikzpicture}
      \caption{Automaton $\B_{m_A,3,(2,2,2)}$ of Theorem~\ref{thm:2exp:improved}}
      \label{fig:B:improved}
    \end{figure}
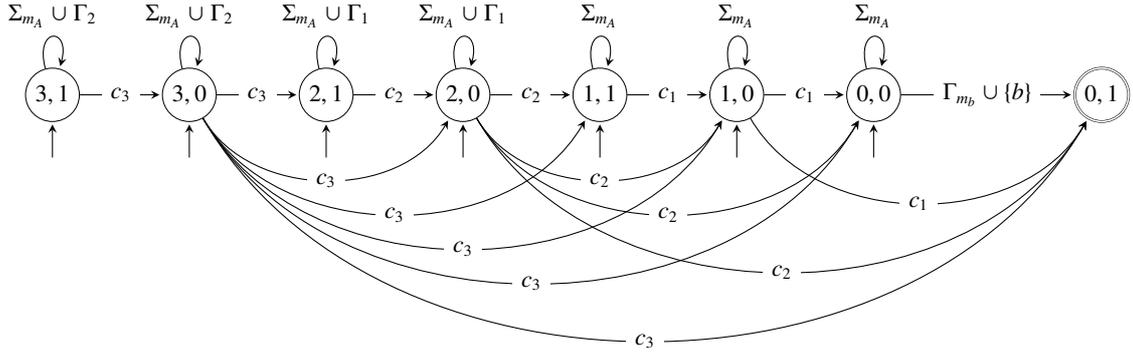
    
    We now define a word $u_{m_A+m_B} b$ such that all its prefixes form a tower. To do this, we first inductively defined a word $\uu_{m_A}$ so that $u_0 = \eps$ and $\uu_{k} = (\uu_{k-1} b a_{k})^{d_k} \uu_{k-1}$, for $0 < k \le m_A$. Then $\uu_k \in \Sigma_k^*$ and $|\uu_k| = 2 \prod_{i=1}^k(d_i+1) - 2$.
    We now use the word $\uu_{m_A}$ to define, for $0<k\leq m_B$, a word $u_{m_A+k}=(u_{m_A+k-1} c_{k})^{e_{k}}u_{m_A+k-1}$. Then $u_{m_A+k}\in (\Sigma_{m_A}\cup \Gamma_{k})^*$ and $|u_{m_A+k}| = \prod_{i=1}^k(e_i+1)\zav{2 \prod_{i=1}^{m_A}(d_i+1)-1}-1$. Finally, the word $u_{m_A+m_B}b$ is of length $\prod_{i=1}^{m_B}(e_i+1)\zav{2 \prod_{i=1}^{m_A}(d_i+1)-1}$, therefore it has $\prod_{i=1}^{m_B}(e_i+1)\zav{2 \prod_{i=1}^{m_A}(d_i+1)-1}+1$ prefixes.
  
    We now show that the prefixes of $u_{m_A+m_B}b$ form a tower between the languages. Namely, we show by induction on $k$ that every prefix $v$ of $u_{m_A+k}b$ is accepted  by $\A_{m_A,m_B,\vek d}$ from a state smaller than $(k+1,0)$ if it ends with a letter from $\Sigma_{m_A+k}\setminus\{b\}$, and it is accepted by $\B_{m_A,m_B,\vek e}$ from an initial state smaller than $(k+1,0)$ if it ends with a letter from $\Gamma_{k}\cup\{b\}$ .

    For $k=0$, it is easy to see that any prefix $v$ of $u_{m_A}b$ ending with $b$ is accepted by $\B_{m_A,m_B,\vek e}$ from state $(0,0)$. We show that if $v$ is a prefix of $u_{m_A}$ ending with a letter from $\Sigma_{m_A}\setminus\{b\}$, then $v$ is accepted by $\A_{m_A,m_B,\vek d}$ from a state smaller than $(k+1,0)$. If $v=u_0=\eps$, then $v$ is accepted by $\A_{m_A,m_B,\vek e}$ from state $(0,0)$. Let $v = (u_{k-1} b a_{k})^{\ell} v'$, for some $v'$ being a prefix of $u_{k-1}$ ending with a letter from $\Sigma_{m_A}\setminus\{b\}$. By the induction hypothesis, $v'$ is accepted by $\A_{m_A,m_B,\vek e}$ from some state $s$ smaller than $(k,0)$. Then $v$ is accepted by $\A_{m_A,m_B,\vek e}$ by the path
    \[
      \underbrace{(k,\ell-1) \xrightarrow{~\uu_{k-1} b~} (k,\ell-1) 
        \xrightarrow{~a_k~} (k,\ell-2)  
        \ \dots\ 
        (k,0) \xrightarrow{~\uu_{k-1} b~} (k,0)
        \xrightarrow{~a_k~} s}_{\ell-\text{times}}
        \xrightarrow{~v'~} (0,0)\,.
    \]

    Thus, let $k\geq 1$ and consider the word $u_{m_A+k}b$. The claim holds for prefixes of $u_{m_A+k-1}b$ by induction. Let $u_{m_A+k-1}$ be accepted in $\A_{m_A,m_B,\vek d}$ from a state $t$. Let $v$ be a prefix of $u_{m_A+k}b$ of the form $(u_{m_A+k-1} c_{k})^\ell v'$, where $1 \leq \ell\leq e_k$ and $v'$ is a prefix of $u_{m_A+k-1}$.
    
    If $v'$ ends with a letter from $\Sigma_{m_A}\setminus\{b\}$, then $\A_{m_A,m_B,\vek d}$ accepts $v'$ from a state $s$ by the induction hypothesis. Then the whole prefix $v$ is accepted in $\A_{m_A,m_B,\vek d}$ by the path 
    \[
      \underbrace{t \xrightarrow{~\uu_{m_A+k-1}~} (0,0) 
        \xrightarrow{~c_k~} t  
        \xrightarrow{~\uu_{m_A+k-1}~}  
        \ \dots\ 
        \xrightarrow{~\uu_{m_A+k-1}~} (0,0) 
        \xrightarrow{~c_k~} s}_{\ell-\text{times}}
        \xrightarrow{~v'~} (0,0)\,.
    \]
    
    If $v'$ ends with a letter from $\Gamma_{k}\cup\{b\}$, it is accepted by $\B_{m_A,m_B,\vek e}$ from a state $s'$ smaller than $(k,0)$ by the induction hypothesis. Then the whole prefix $v$ is accepted by $\B_{m_A,m_B,\vek e}$ by the path
    \[
      (k,\ell-1)  \xrightarrow{~\uu_{m_A+k-1}~} (k,\ell-1) 
        \xrightarrow{~c_{k}~} (k,\ell-2) 
        \xrightarrow{~\uu_{m_A+k-1}~} (k,\ell-2) \quad 
        \cdots \quad (k,0) 
        \xrightarrow{~\uu_{m_A+k-1}~} (k,0)
        \xrightarrow{~c_{k}~} s'
        \xrightarrow{~v'~} (0,1)\,.
    \]
    This shows the claimed height of the tower. 

    It remains to show that there is no infinite tower. We first show that there is no infinite tower over the alphabet $\Sigma_{m_A}$ and then that there is no infinite tower over the alphabet $\Sigma_{m_A}\cup\Gamma_{m_B}$. Suppose the contrary, and let $k \geq 0$ be the smallest integer such that there is an infinite tower over $\Sigma_k$. Since $\eps,b$ is the highest tower over $\Sigma_0$, we have $k\ge 1$. Since every word of $L(\A_{m_A,m_B,\vek d})\cap \Sigma_k^*$ contains at most $d_k$ occurrences of the letter $a_k$, we can consider, without loss of generality, an infinite tower $(w_i)_{i=1}^{\infty}$ in which every $w_i$ contains the same (nonzero) number of occurrences of $a_k$. Let $w_i = w_i'' a_k w_i'$, where $w_i' \in \Sigma_{k-1}^*$. Since all transitions under $a_k$ lead to an initial state, the word $w_i'$ is accepted by $\A_{m_A,m_B,\vek d}$ (by $\B_{m_A,m_B,\vek e}$ resp.) if $w_i$ is accepted by $\A_{m_A,m_B,\vek d}$ (by $\B_{m_A,m_B,\vek e}$ resp.). Then $(w_i')_{i=1}^\infty$ is an infinite tower over $\Sigma_{k-1}$; a contradiction.
    Thus, let $k\ge 1$ be the smallest integer such that there is an infinite tower over $\Sigma_{m_A}\cup \Gamma_k$. Similarly as above, since the number of occurrences of $c_k$ in words from $(\Sigma_{m_A}\cup \Gamma_k)^*$ accepted by the automaton $\B_{m_A,m_B,\vek d}$ is restricted to $e_k$, we can choose an infinite tower $(w_i)_{i=1}^\infty$ such that all words have the same (nonzero) number of occurrences of the letter $c_k$. Let $w_i = w_i'' c_k w_i'$ with $w_i' \in (\Sigma_{m_A} \cup \Gamma_{k-1})^*$. A direct inspection of the automata yields that $w_i'$ is accepted by $\A_{m_A,m_B,\vek d}$ (by $\B_{m_A,m_B,\vek e}$ resp.) if $w_i$ is accepted by $\A_{m_A,m_B,\vek d}$ (by $\B_{m_A,m_B,\vek e}$ resp.). Then $(w_i')_{i=1}^\infty$ is an infinite tower over $\Sigma_{m_A} \cup \Gamma_{k-1}$; a contradiction.
  \end{proof}

  The following corollaries, improving the previous results, are straightforward.

  \begin{cor}\label{cor:A1}
    For any integers $n_1,n_2\ge 2$, there exist two NFAs with $n_1$ and $n_2$ states over an alphabet of cardinality $n_1+n_2-2$ having a tower of height $2^{n_1+n_2-2}-2^{n_2-2}+1$ and no infinite tower.
  \end{cor}
  \begin{proof}
    Let  $m_A=n_1-1$, $m_B=n_2-2$, and $d_i=e_j=1$, for all $1\leq i\leq m_A$ and $1\leq j\leq m_B$. The claim now follows from Theorem~\ref{thm:2exp:improved}.
  \end{proof}

   Choosing $n_2=2$, we have the following result.
  \begin{cor}\label{cor:A2}
    For every integer $n\ge 2$, there exist an NFA with $n$ states and a DFA with two states over an alphabet of cardinality $n$ having a tower of height $2^n$ and no infinite tower.
  \end{cor}
  \begin{proof}
    Let $m\ge 1$ and set $m_A=m$, $m_B=0$, and $d_i=1$, for all $i$. By Theorem~\ref{thm:2exp:improved}, there are two NFAs with $n=m+1$ and $2$ states, respectively, over an alphabet of cardinality $m+1$ having a tower of height $2^{m+1}$ and no infinite tower. Notice that the automaton $\B_{m_A,m_B,\vek e}$ for $m_B=0$ has two states and its deterministic counterpart as well (cf. the automaton $\B_n$ in Figure~\ref{fig4a}). This completes the proof.
  \end{proof}

\section{Lower bounds on the height of towers for DFAs}
  Exponential lower bounds presented above are based on NFAs. It is an interesting question whether they can also be achieved for DFAs. We answer this question in this section.

  \subsection{Lower bounds on the height of towers for DFAs over a growing alphabet} 
  \begin{thm}\label{thm:expdfa}
    For every $n\ge 0$, there exist two DFAs with at most $n+1$ states over an alphabet of cardinality $\frac{n(n+1)}{2}+1$ having a tower of height $2^{n}$ and no infinite tower.
  \end{thm}
  \begin{proof}
    The main idea of the construction is to ``determinize'' the automata of the proof of Theorem~\ref{thm:2exp:improved} with $m_B=0$. For the sake of simplicity, we consider the case $m=n$ and $d_i=1$, for $i=1,2,\dots,n$, and label states $(i,0)$ simply as $i$.
		
		For a given integer $n$, we define a pair of deterministic automata $\A_{n}$ and $\B_{n}$ with $n+1$ and two states, respectively, over the alphabet $\Sigma_n=\{b\}\cup\{a_{i,j} \mid i=1,2,\dots,n;\, j=0,1,\dots,i-1\}$ with a tower of height $2^{n}$ between $L(\A_n)$ and $L(\B_n)$, and with no infinite tower.
    The two-state DFA $\B_n=(\{1,2\},\Sigma_n,\gamma_n,1,\{2\})$ accepts all words over $\Sigma_n$ ending with $b$ and is shown in Figure~\ref{fig4a} (right).
    \begin{figure}
      \centering
      \begin{tikzpicture}[baseline,->,>=stealth,shorten >=1pt,node distance=1.8cm,
        state/.style={circle,minimum size=6mm,very thin,draw=black,initial text=},
        every node/.style={fill=white,font=\small},
        bigloop/.style={shift={(0,0.05)},text width=1.6cm,align=center}]
        \node[state,accepting]    (1) {$0$};
        \node[state]              (4) [left of=1] {$1$};
        \node[state]              (5) [left of=4] {$2$};
        \node[state,initial]              (6) [left of=5] {$3$};
        \path
          (4) edge node {$a_{1,0}$} (1)
          (4) edge[loop above] node[bigloop] {$b$\\$a_{2,0},a_{3,0}$} (4)
          (5) edge[loop above] node[bigloop] {$b,a_{1,0},$\\$a_{3,0},a_{3,1}$} (5)
          (6) edge[loop above] node[bigloop] {$b,a_{1,0},$\\$a_{2,0},a_{2,1}$} (6)
          (5) edge[bend right=55] node {$a_{2,0}$} (1)
          (5) edge node {$a_{2,1}$} (4)
          (6) edge node {$a_{3,2}$} (5)
          (6) edge[bend right=55] node {$a_{3,1}$} (4)
          (6) edge[bend right=60] node {$a_{3,0}$} (1) ;
      \end{tikzpicture}
      \qquad\qquad
      \begin{tikzpicture}[baseline,->,>=stealth,shorten >=1pt,node distance=2.3cm,
        state/.style={circle,minimum size=6mm,very thin,draw=black,initial text=},
        every node/.style={font=\small}]
        \node[state,initial]    (1) {$1$};
        \node[state,accepting]  (2) [right of=1] {$2$};
        \path
          (1) edge[loop above] node {$\Sigma_n\setminus\{b\}$} (1)
          (2) edge[loop above] node {$b$} (2)
          (1) edge[bend left] node[fill=white] {$b$} (2)
          (2) edge[bend left] node[fill=white] {$\Sigma_n\setminus\{b\}$} (1);
      \end{tikzpicture}
      \caption{The DFA $\A_3$ (left) and the two-state DFA $\B_n$ (right), $n\ge 0$.}
      \label{fig4a}
      \label{fig5a}
      \label{figA3var}
    \end{figure}
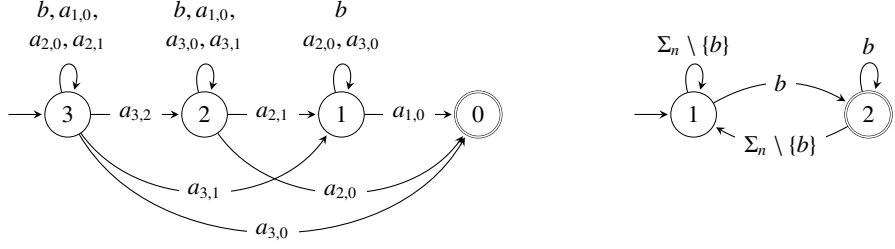
    \begin{figure}[b]
      \centering
      \begin{tikzpicture}[->,>=stealth,shorten >=1pt,node distance=1.8cm,
        state/.style={circle,minimum size=6mm,very thin,draw=black,initial text=},
        every node/.style={fill=white,font=\small},
        bigloop/.style={shift={(0,0.05)},text width=1.6cm,align=center}]
        \node[state,accepting]    (1) {$0$};
        \node[state]              (4) [left of=1] {$1$};
        \node[state]              (5) [left of=4] {$2$};
        \node[state]              (6) [left of=5] {$3$};
        \node[state, initial]     (7) [left of=6] {$4$};
        \path
          (4) edge node {$a_{1,0}$} (1)
          (4) edge[loop above] node[bigloop] {$b$\\$a_{2,0},a_{3,0},a_{4,0}$} (4)
          (5) edge[loop above] node[bigloop] {$b$\\$a_{3,1},a_{4,1},$\\$a_{1,0},a_{3,0},a_{4,0}$} (5)
          (6) edge[loop above] node[bigloop] {$b$\\$a_{4,2}$\\$a_{2,1},a_{4,1},$\\$a_{1,0},a_{2,0},a_{4,0}$} (6)
          (7) edge[loop above] node[bigloop] {$b$\\$a_{3,2}$\\$a_{2,1},a_{3,1},$\\$a_{1,0},a_{2,0},a_{3,0}$} (7)
          (5) edge[bend right=55] node {$a_{2,0}$} (1)
          (5) edge node[fill=white] {\small{$a_{2,1}$}} (4)
          (6) edge node {$a_{3,2}$} (5)
          (6) edge[bend right=55] node {$a_{3,1}$} (4)
          (6) edge[bend right=60] node {$a_{3,0}$} (1)
          (7) edge node {$a_{4,3}$} (6)
          (7) edge[bend right=55] node {$a_{4,2}$} (5)
          (7) edge[bend right=60] node {$a_{4,1}$} (4)
          (7) edge[bend right=65] node {$a_{4,0}$} (1) ;
      \end{tikzpicture}
      \caption{Automaton $\A_4$.}
      \label{figA4var}
    \end{figure}
    
    The ``determinization'' idea of the construction of the DFA $\A_n=(\{0,1,\dots,n\},\Sigma_n,\delta_n,n,\{0\})$ is to use the automaton $\A_{n,(1,\dots,1)}$ from the proof of Theorem~\ref{thm:2exp:improved}, and to eliminate the nondeterminism by relabeling every transition $(i,0) \xrightarrow{a_i} (j,0)$ with a new unique letter $(i,0) \xrightarrow{a_{i,j}} (j,0)$. Then the tower of Theorem~\ref{thm:2exp:improved} is modified by relabeling the corresponding letters. However, to preserve embeddability of the new letters, several self-loops must be added. 
		
		Formally, the transition function $\delta_n$ is defined as follows.
    For every $a_{i,j}\in \Sigma_n$, we define the transition $\delta_n(i,a_{i,j}) = j$.
    For every $k = 1,2,\dots,n$ and $a_{i,j}\in \Sigma_n$ such that $i\neq k$ and $j < k$, we define the self-loop $\delta_n(k, a_{i,j}) = k$.
    Finally, we add the self-loops $\delta_n(k,b) = k$ to every state $k=1,2,\dots,n$, see Figures~\ref{figA3var} and~\ref{figA4var} for an illustration. 
    
    For every $1\le k\le n$ and $0\le j < k$, let $\myalpha k j = a_{k,{j}}a_{k,j-1}\cdots a_{k,0}$, and let the words $u_k$ be defined by $u_0=\varepsilon$ and $u_{k}=u_{k-1}b\, \myalpha{k}{k-1}\, u_{k-1}$. Note that $u_kb$ contains $2^k$ letters $b$.

    We first give an informal description of the tower of height $2^n$ between $\A_n$ and $\B_n$, which relates the construction to Theorem~\ref{thm:2exp:improved}. The tower is the sequence $w_n(0), w_n(1), \dots, w_n(2^n-1)$, where the longest word is defined by 
    \[
      w_n=w_n(2^n-1)=\myalpha{n}{n-1}\,u_{n-1}b\in L(\B_n)\,.
    \]
    The word $w_n(2i)$ is obtained from the word $w_n(2i+1)$ by removing the last letter, which is $b$. The word $w_n(2i-1)$ is obtained from the word $w_n(2i)$ by removing the first letter of some occurrences of $\myalpha k j$ in $w_n(2i)$, see Figure~\ref{figTower} for the case $n=3$. 
    \begin{figure}
      \centering
      \begin{align*}
        w_3(0)&=\underline{a_{3,0}}\\
        w_3(1)&=\underline{a_{3,0}}\,b\\
        w_3(2)&=\underline{a_{3,1}}a_{3,0}\,b\,\underline{a_{1,0}}\\
        w_3(3)&=\underline{a_{3,1}}a_{3,0}\,b\,\underline{a_{1,0}}\,b\\
        w_3(4)&=\underline{a_{3,2}}a_{3,1}a_{3,0}\,b\,a_{1,0}\,b\,\underline{a_{2,0}}\\
        w_3(5)&=\underline{a_{3,2}}a_{3,1}a_{3,0}\,b\,a_{1,0}\,b\,\underline{a_{2,0}}\,b\\
        w_3(6)&=\underline{a_{3,2}}a_{3,1}a_{3,0}\,b\,a_{1,0}\,b\,\underline{a_{2,1}}a_{2,0}\,b\,\underline{a_{1,0}}\\
        w_3(7)&=\underline{a_{3,2}}a_{3,1}a_{3,0}\,b\,a_{1,0}\,b\,\underline{a_{2,1}}a_{2,0}\,b\,\underline{a_{1,0}}\,b
      \end{align*}
      \caption{The tower between $L(\A_3)$ and $L(\B_3)$. We underline transitions between different states in $\A_3$.}
      \label{figTower}
    \end{figure}
  
   We now give a formal definition of $w_n(i)$, which is done recursively. For any $k\geq 1$, we define
    $w_k(0)=\myalpha k 0=a_{k,0}$ and
    $w_k(1)=a_{k,0}\, b$.
    For $2 \leq i\leq 2^k-1$, let  
    \begin{align}\label{eq:wki}
      w_k(i)=\myalpha k \jjj\, u_{\jjj-1}\,b\, w_{\jjj}\left(i-2^{\jjj}\right).    
    \end{align}

    By double induction on $n$ and $i$, we prove that the sequence $\left(w_n(i)\right)_{i=0}^{2^n-1}$ is the required tower. For $n=1$, the claim holds, the tower is $w_1(0)=a_{1,0}$, $w_1(1)=a_{1,0}\,b$. Let $n>1$. The definition implies, by induction, that $w_n(i)$ is in $L(\B_n)$ (that is, it ends with $b$) if and only if $i$ is odd. Consider $w_n(i)$ with even $i \ge 2$. 
    Using~\eqref{eq:wki}, there is a path in $\A_n$ labeled by $w_n(i)$ and it can be decomposed as
    \[
      n \xrightarrow{~a_{n,\jjj}~} \jjj
        \xrightarrow{~{\myalpha{n}{\jjj-1}}\,u_{\jjj-1}\,b~} \jjj
        \xrightarrow{~w_{\jjj}\left(i-2^{\jjj}\right)~} 0\,.
    \]
    For the second segment of the path, note that both the alphabet of ${\myalpha{n}{\jjj-1}}$ and the alphabet $\{b\}\cup \{a_{m,m'} \mid m\leq \jjj-1, m' < m\}$ of $u_{\jjj-1}\, b$ are contained in the alphabet of self-loops of state $\jjj$. The last segment exists by induction, since $\jjj<n$,  $i-2^\jjj\leq 2^\jjj-1$, the automaton $\A_\jjj$ is a restriction of $\A_n$, and $i-2^\jjj$ is even.

   We show that $w_n(i)\preq w_n(i+1)$. This is true for $i=0$, and follows by induction from \eqref{eq:wki} if $\floor{\log (i+1)}=\floor{\log i}$. The latter equality holds unless  $i$ is of the form $2^\ell-1$ for some $\ell > 1$. If $i=2^\ell-1$, then $\ell-1=\floor{\log i}\neq \floor{\log (i+1)}=\ell$ and we have
    \begin{align*}
      \begin{split}
        w_n\left(i\right) & = \myalpha{n}{\ell-1}\,u_{\ell-2}\,b\,w_{\ell-1}\left(2^{\ell-1}-1\right) 
                                = \myalpha{n}{\ell-1}\,u_{\ell-2}\,b\,\myalpha{\ell-1}{\ell-2}\,u_{\ell-2}\,b 
                                = \myalpha{n}{\ell-1}\,u_{\ell-1}\,b,\\
        w_n\left(i+1\right)   & =\myalpha n \ell\,u_{\ell-1}\,b\, a_{\ell,0}\,,
      \end{split}
    \end{align*}
    hence  $w_n(i)\preq w_n(i+1)$ holds.
    
		Finally, observe that if $\left(v_i\right)$ is a tower between $\A_n$ and $\B_n$, then $\left(P(v_i)\right)$ is a tower between $\A_{n,(1,\dots,1)}$ and $\B_n$, where $P\colon a_{k,j} \mapsto a_k$ is the natural projection of $\Sigma_n$ to $\Sigma$. Therefore there is no infinite tower between $\A_n$ and $\B_n$ by Theorem~\ref{thm:2exp:improved}.
  \end{proof}

  The ``determinization'' idea of the previous theorem can be generalized. However, compared to the proof of Theorem~\ref{thm:expdfa}, the general procedure suffers from the increase of states. The reason why we need not increase the number of states in the proof of Theorem~\ref{thm:expdfa} is that the automata we are ``determinizing'' are such that there is an order in which the transitions/states are used/visited, and that the nondeterministic transitions are acyclic.

  \begin{thm}\label{determinisation}
    For every two NFAs $\A$ and $\B$ with at most $n$ states and $m$ input letters, there exist two DFAs $\A'$ and $\B'$ with $O\zav{n^2}$ states and $O\zav{m+n}$ input letters such that there is a tower of height $r$ between $\A$ and $\B$ if and only if there is a tower of height $r$ between $\A'$ and $\B'$. In particular, there is an infinite tower between $\A$ and $\B$ if and only if there is an infinite tower between $\A'$ and $\B'$.
  \end{thm}
  \begin{proof}
    Let $\A$ and $\B$ be two NFAs with at most $n$ states over an alphabet $\Sigma$ of cardinality $m$. Without loss of generality, we may assume that the automata each have a single initial state. Let $Q_A$ and $Q_B$ denote their respective sets of states. We modify the automata $\A$ and $\B$ to obtain the DFAs $\A'$ and $\B'$ as follows. Let $Q_{A'} = Q_A \cup \{ \sigma_{s,t} \mid s, t \in Q_A \}$ and $Q_{B'} = Q_B \cup \{ \sigma_{s,t} \mid s,t \in Q_B \}$, where $\sigma_{s,t}$ are new states. We introduce a new letter $y_{t}$ for every state $t \in Q_A \cup Q_B$. It results in $O\zav{n^2}$ states and $O\zav{m+n}$ letters. The transition function is defined as follows. In both automata, each transition $s \xrightarrow{a} t$ is replaced with two transitions $s \xrightarrow{y_{t}} \sigma_{s,t}$ and $\sigma_{s,t} \xrightarrow{a} t$. Moreover, self-loops in all new states are added over all new letters. Note that all transitions are deterministic in $\A'$ and $\B'$.
    
    We now prove that if there is a tower of height $r$ between $\A$ and $\B$, then there is a tower of height $r$ between $\A'$ and $\B'$.
    Let $\left(w_i\right)_{i=1}^{r}$ be a tower between $\A$ and $\B$. Let
    \[
      w_i=x_{i,1}x_{i,2}\cdots x_{i,n}\,,
    \]
    where $n=|w_r|$ and $x_{i,j}$ is either a letter or the empty word such that $x_{i,j} \preq x_{i+1,j}$, for each $i=1,2,\dots,r-1$ and $j=1,2,\dots,n$.
    For every $w_i$, we fix an accepting path $\pi_i$ in the corresponding automaton.
    Let $p_{i,j}$ be the letter $y_{t}$ where $s \xrightarrow{} t$ is the transition corresponding to $x_{i,j}$ in $\pi_i$ if $x_{i,j}$ is a letter, 
    and let $p_{i,j}$ be empty if $x_{i,j}$ is empty. 
    We define 
    \[
      w_i'=\myalpha{i}{1}\myalpha{i}{2}\cdots \myalpha{i}{n}\,,
    \]
    where $\myalpha{i}{j} = p_{i,j} p_{i-1,j} \cdots p_{1,j} x_{i,j}$ if $x_{i,j} \neq \varepsilon$, and $\myalpha{i}{j} = x_{i,j} = \varepsilon$ otherwise. 
    It is straightforward to verify that $\left(w_i'\right)_{i=1}^r$ is a tower of height $r$ between $\A'$ and $\B'$.
    
    Let now $\left(w_i'\right)_{i=1}^r$ be a tower between $\A'$ and $\B'$. We show that $\left(P\left(w_i'\right)\right)_{i=1}^r$ is a tower between $\A$ and $\B$, where $P$ is a projection erasing all new letters. Obviously, we have $P\left(w_i'\right) \preq P\left(w_{i+1}'\right)$.
		We now show that if a word $w'$ is accepted by $\A'$, then $P\left(w'\right)$ is accepted by $\A$.  
    Let $\pi'$ be the path accepting $w'$, and let $\tau_1' , \tau_2' , \ldots, \tau_k'$ denote the sequence of all transitions of $\pi'$ labeled with letters from $\Sigma$ in the order they appear in $\pi'$. 
    By construction, $\tau_i'$ is of the form $\sigma_{s_{i-1},s_{i}} \xrightarrow{a_i} s_{i}$ for some states $s_i \in Q_A$, $i=0,2,\dots,k$, with $s_0$ being initial and $s_k$ being accepting. Moreover, for $i<k$, the transition $\tau'_i$ is immediately followed in $\pi'$ by $s_i \xrightarrow{y_{s_{i+1}}} \sigma_{s_i,s_{i+1}}$.
    Let $\tau_i$ be $s_{i-1} \xrightarrow{a_i} s_i$. 
		It is straightforward to verify that $\tau_1, \tau_2, \dots, \tau_k$ is an accepting path of $p(w')$ in $\A$. Analogously for $\B'$ and $\B$. As the existence of towers of arbitrary height is equivalent to the existence of an infinite tower, this concludes the proof.
  \end{proof}

  A similar construction yields the following variant of the previous theorem.
  \begin{thm}\label{determinisation2}
    For every two NFAs $\A$ and $\B$ with at most $n$ states and $m$ input letters, there exist two DFAs $\A'$ and $\B'$ with $O(mn)$ states and $O(m n)$ input letters such that there is a tower of height $r$ between $\A$ and $\B$ if and only if there is a tower of height $r$ between $\A'$ and $\B'$. In particular, there is an infinite tower between $\A$ and $\B$ if and only if there is an infinite tower between $\A'$ and $\B'$.
  \end{thm}
  \begin{proof}
    Let $Q_{A'} = Q_A \cup \{ \sigma_{a,t} \mid a \in \Sigma,\, t \in Q_A \}$ and $Q_{B'} = Q_B \cup \{ \sigma_{a,t} \mid a \in \Sigma,\, t \in Q_B \}$, where $\sigma_{a,t}$ are new states. The alphabet of $\A'$ and $\B'$ is $\Sigma \cup \{ a_t \mid a \in \Sigma, t \in Q_A \cup Q_B  \}$. We have $O(mn)$ states and letters. Each transition $s \xrightarrow{a} t$, in both automata, is replaced with two transitions $s \xrightarrow{a_{t}} \sigma_{a,t}$ and $\sigma_{a,t} \xrightarrow{a} t$. Self-loops in all new states are added over all new letters.
    The rest of the proof is analogous to the proof of Theorem~\ref{determinisation}.
  \end{proof}

\subsection{Lower bounds on the height of towers for DFAs over a fixed alphabet}
  The size of the alphabet in the previous constructions depends on the number of states, hence these constructions cannot be used to answer the questions whether the upper bound of Theorem~\ref{thm01} is tight for DFAs over a fixed alphabet. We answer this question now.

  \begin{thm}\label{thm:lower_bound:dfa}
    For all integers $m$, $d_1,d_2,\dots, d_m \ge 2$ and every odd positive integer $e$, there exist two DFAs with $2 d_1 + \sum_{i=2}^m (d_i+1)$ and $e + 1$ states over an alphabet of cardinality $m + 1$ having a tower of height $\left((e+1)d_{1}+2\right)\prod_{i=2}^m d_i$ and no infinite tower.
  \end{thm}
  \begin{proof}
    For every $m\ge 1$, we define the alphabet $\Sigma_m = \{b,a_1,a_2,\ldots,a_m\}$. We set $\vek{d}=(d_1,\dots,d_m)$, and define two DFAs $\A_{m,\vek{d}}$ and $\B_{e}$ over $\Sigma_m$ as follows.
    The set of states of the DFA $\A_{m,\vek d}$ is $Q_{m,\vek d} = \{ (k,j) \mid k=1,2,\dots,m; \allowbreak \ j=0,1,2,\dots,d_k-1\} \cup \{k' \mid k=2,\dots,m\} \cup \{ (1,i') \mid i=0,1,\ldots,d_1-1 \}$. State $(m,d_{m}-1)$ is initial, and states $(1,j)$, $j=0,\ldots,d_1-1$, are accepting. The transition function of $\A_{m,\vek d}$ consists of

    \begin{itemize}
      \itemsep0pt
      \item  an $a_i$-transition from each $(i,j)$ to $(i,j-1)$, and 
        from each $(i,0)$ to states $(i-1,d_{i-1}-1)$, for $i> 1$ and $j>0$,
      \item an $a_1$-transition from each $(1,j')$ to $(1,j-1)$,  for $1\leq j \leq d_1-1$,
			\item a transition from $(i,d_i-1)$ to $i'$ and back under $a_1$, for $i\ge 2$,
      \item  a $b$-transition from $(1,j)$ to $(1,j')$ and back, for $0\le j\le i$,
      \item  self-loops in $(i,j)$, for $i\ge 2$ and $0\le j < d_i-1$, under $\Sigma_{i-1}$, and
      \item self-loops under $a_i$ in states $i'$, for $i\ge 2$.
    \end{itemize}
    See Figure~\ref{DFA_A} for an example.
  
    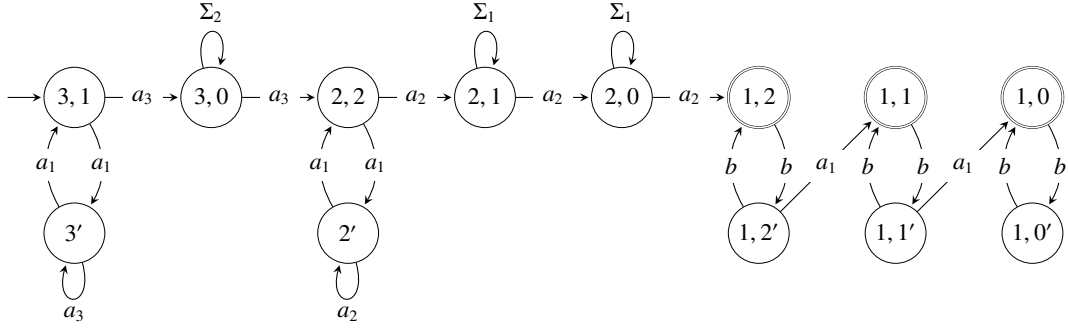
\begin{figure}[h]
      \centering
      \begin{tikzpicture}[baseline,->,>=stealth,shorten >=1pt,node distance=1.8cm,
        state/.style={circle,minimum size=8mm,very thin,draw=black,inner sep=2pt,initial text=},
        every node/.style={font=\small},
        bigloop/.style={shift={(0,0.01)},text width=1.6cm,align=center}]
        \node[state,accepting]  (00)  {$1,0$};
        \node[state]            (00') [below of=00] {$1,0'$};
        \node[state,accepting]  (10)  [left of=00]  {$1,1$};
        \node[state]            (10') [below of=10] {$1,1'$};
        \node[state,accepting]  (11)  [left of=10]  {$1,2$};
        \node[state]            (11') [below of=11] {$1,2'$};
        \node[state]            (20)  [left of=11]  {$2,0$};
        \node[state]            (21)  [left of=20]  {$2,1$};
        \node[state]            (22)  [left of=21]  {$2,2$};
        \node[state]            (22') [below of=22] {$2'$};
        \node[state]            (30)  [left of=22]  {$3,0$};
        \node[state,initial]    (31)  [left of=30]  {$3,1$};
        \node[state]            (31') [below of=31] {$3'$};
        
        \foreach \from/\to in {00/00',00'/00,10/10',10'/10,11/11',11'/11}
          \path (\from) edge[bend left] node[fill=white] {$b$} (\to);
        
        \foreach \from/\to in {22/22',22'/22,31/31',31'/31}
          \path (\from) edge[bend left] node[fill=white] {$a_1$} (\to);
        
        \foreach \from/\to/\pis in {31/30/3,30/22/3,22/21/2,21/20/2,20/11/2,11'/10/1,10'/00/1}
          \path (\from) edge node[fill=white] {$a_\pis$} (\to);
        
        \foreach \from/\pis in {30/{$\Sigma_2$},21/{$\Sigma_1$},20/{$\Sigma_1$}}
          \path (\from) edge[loop above] node[bigloop] {\pis} (\from);
          
        \foreach \from/\pis in {22'/2,31'/3}
          \path (\from) edge[loop below] node[bigloop] {$a_\pis$} (\from);
      \end{tikzpicture}
      \caption{Automata $\A_{3,(2,3,2)}$ of Theorem~\ref{thm:lower_bound:dfa}}
      \label{DFA_A}
    \end{figure}

    The DFA $\B_e$ has the states $\{0,1,\dots,e\}$. State $0$ is initial and states with odd numbers are accepting. The transitions of $\B_e$ contain 
      $b$-transitions from state $i$ to state $i+1$, for each $0\leq i\leq e-1$,
      a transition under $\{a_1,a_2,\dots,a_m\}$ from state $e$ to state $0$, 
      and a self-loop in state $0$ under $\{a_1,a_2,\dots,a_m\}$, see Figure~\ref{DFA_B} for an illustration.
    \begin{figure}
      \centering
      \begin{tikzpicture}[baseline,->,>=stealth,shorten >=1pt,node distance=1.6cm,
        state/.style={circle,minimum size=1mm,very thin,draw=black,initial text=},
        every node/.style={fill=white,font=\small},
        bigloop/.style={shift={(0,0.01)},text width=1.6cm,align=center}]
        \node[state,initial]    (1) {$0$};
        \node[state,accepting]  (2) [right of=1] {$1$};
        \node[state]            (3) [right of=2] {$2$};
        \node[state, accepting] (4) [right of=3] {$3$};
        \node[state]            (5) [right of=4] {$4$};
        \node[state, accepting] (6) [right of=5] {$5$};
        
        \foreach \from/\to in {1/2,2/3,3/4,4/5,5/6}
          \path (\from) edge node {$b$} (\to);
          
        \path (6) edge[bend right] node {$\Sigma_m\setminus\{b\}$} (1);
        
        \path (1) edge[loop above] node[bigloop] {$\Sigma_m\setminus\{b\}$} (1);
        
       \end{tikzpicture}
      \caption{Automata $\B_{5}$ of Theorem~\ref{thm:lower_bound:dfa}}
      \label{DFA_B}
    \end{figure}
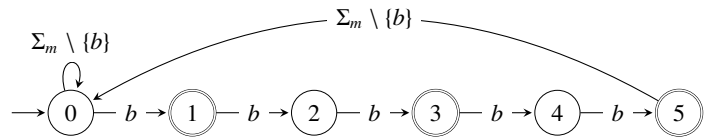

    We now show by induction on $m$ that there is a tower of height $\left(d_1(e+1)+2\right)\prod_{i=2}^{m} d_i$ between $\A_{m,\vek{d}}$ and $\B_e$.
    Note that for $m=1$, the automata $\A_{1,(d_1)}$ and $\B_e$ are (deterministic) variants of automata from Theorem \ref{thm02B}, and  they have the tower of prefixes of the word $u=\left(b^ea_1\right)^{d_1-1}b^e$ with two additional words: $ub$ accepted by $\A_{1,(d_1)}$, and $ua_1b^e$ accepted by $\B_e$ in the state $e$. 
    The tower has the required height $d_1(e+1)+2$.
     (cf. also Theorem \ref{thm:dfas:tight} and Figure \ref{exdfa} below). 

    Let now 
    $
      (w_{m-1,i})_{i=1}^{\ell},
    $
    with $\ell=(d_1(e+1)+2)\prod_{j=2}^{m-1} d_j$, be a tower between the automata $\A_{m-1,(d_1,\ldots,d_{m-1})}$ and $\B_e$, where $w_{m-1,1}$ is accepted by $\A_{m-1,(d_1,\ldots,d_{m-1})}$ and $w_{m-1,\ell}$ is accepted by $\B_e$ in the state $e$.
		Let $u_{m,j}=(a_1a_ma_1)^{j}a_m^{d_k-j}\left(w_{m-1,\ell} a_m\right)^{j}$. We show that $(w_{m,i})_{i=1}^{d_m\ell}$ with
		\[
		w_{m,j\ell+k}=u_{m,j}w_{m-1,k},\quad j=0,1,\dots,d_m-1, \ k=1,2,\dots,\ell,
		\] 
		is a tower between the automata $\A_{m,(d_1,\ldots,d_{m})}$ and $\B_e$, where $w_{m,1}$ is accepted by $\A_{m,(d_1,\ldots,d_{m})}$ and $w_{m,d_m\ell}$ is accepted by $\B_e$ in the state $e$.
		
		Since $w_{m-1,\ell}\in \Sigma_{m-1}^*$, it is straightforward to verify that, for each $j=0,\dots,d_m-1$, the word $u_{m,j}$ is a label, in $\A_{m,(d_1,\ldots,d_{m})}$, for a transition from $(m,d_m-1)$ to $(m-1,d_{m-1}-1)$, namely
		\[
      (m,d_m-1) \xrightarrow{~(a_1 a_m a_1)^j ~} (m,d_m-1) 
                \xrightarrow{~a_m^{d_m-j}~} (m,j-1) 
                \xrightarrow{~(w_{m-1,\ell}a_m)^j~} (m-1,d_{m-1}-1) 
              \,.
    \]
		In $\B_e$, each $u_{m,j}$ is a label for a cycle starting and ending in the state $0$. Namely
		\[
      0 \xrightarrow{~(a_1 a_m a_1)^ja_m^{d_m-j}~} 0 
                \underbrace{\xrightarrow{~w_{m-1,\ell}~} e
								\xrightarrow{~a_m~} 0}_{j-\text{times}}	
              \,.
    \]
		
		This implies, by induction and by the construction of the automata, that $w_{m,i}$ are accepted as required. By induction, we have that $w_{m,j\ell+k}\preq w_{m,j\ell+k+1}$ for each $1\leq k<\ell$, and the definition of $u_{m,j}$ implies that also 
		\[
			w_{m,j\ell}= (a_1a_ma_1)^{j-1}a_m^{d_k-j+1}\left(w_{m-1,\ell} a_m\right)^{j-1}w_{m-1,\ell} 
			\preq (a_1a_ma_1)^{j}a_m^{d_k-j}\left(w_{m-1,\ell} a_m\right)^{j} =w_{m,j\ell+1}.
		\]
		This completes the proof that  $(w_{m,i})_{i=1}^{d_m\ell}$ is a tower with required properties.

    It remains to show that there is no infinite tower. For the sake of contradiction, let $m$ be the smallest number such that there exists an infinite tower $(w_{i})_{i=1}^{\infty}$ between $\A_{m,\vek d}$ and $\B_e$, for some $\vek d$ and $e$. By  Theorem \ref{thm02B}, we know that $m> 1$. 
				Suppose, first, that for all $i$, $w_i\in \{a_1,a_m\}^*w_i'$ where $w_i'\in \Sigma_{m-1}^*$. 
				We may assume that $a_1$ is not the first letter of $w_i'$.
				This implies that after reading the $\{a_1,a_m\}^*$ part, $\A_{m,\vek{d}}$ is in $(m-1,d_{m-1}-1)$ and that $\B_e$ is in $0$. Thus $(w_i')_{i=1}^\infty$ is an infinite tower between $\A_{m-1,(d_1,\dots,d_{m-1})}$ and $\B_e$; a contradiction.
					Let now $t> 1$ be the largest integer such that a word $ca_m^t\preq w_i$ for some $i$, where $c\in \Sigma_m\setminus \{a_1,a_m\}$. It is straightforward to verify that $ca_m^{d_m}$ cannot be embedded into any word from $L(\A_{m,\vek d})$, hence $t< d_m$. 
		Without loss of generality, we can suppose that $ca_m^t \preq w_i$ for all $i$. Let $w_i=w_i''w_i'$ where $w_i''$ is the shortest prefix of $w_i$, such that $ca_m^t\preq w_i''$. Then $w_i'\in \Sigma_{m-1}^*$, and  $(w_i')_{i=1}^\infty$ is again a tower between $\A_{m-1,(d_1,\dots,d_{m-1})}$ and $\B_e$; a contradiction. 
		
\end{proof}

  As a corollary, we have that the upper bound of Theorem~\ref{thm01} is tight for a fixed alphabet even for DFAs.
  \begin{cor}\label{cor16}  
    Let $k \ge 2$ be a constant. Then the maximum height of a tower between two DFAs with at most $n$ states over an alphabet of cardinality $k$ having no infinite tower is in $\Omega\zav{n^k}$.
  \end{cor}

  \begin{proof}
	Let $m=k-1$, and let $\ell=\left\lfloor \frac{n}{k} \right\rfloor$. For sufficiently large $n$, there is an integer $\ell \leq d_1 \leq \ell+1$, integers $\ell-1 \leq d_i \leq \ell$, $i=2,\dots,m$, and an odd integer  $n-2 \leq e\leq n-1$ such that Theorem~\ref{thm:lower_bound:dfa} yields two DFAs with $n$ states over an alphabet of cardinality $k=m+1$ having a tower of height at least $((n-1)\ell+2)(\ell-1)^{k-2} \in \Omega\zav{n^{k}}$.
\end{proof}

\section{Towers of prefixes}\label{TofPref}
  It is remarkable that lower bounds on the height of finite towers for NFAs in this paper were obtained by examples where $w_{i}$ is not just a subsequence of $w_{i+1}$ but even its prefix (sometimes this rule is violated by the last element of the tower). In this section we therefore investigate what can be said about alternating towers of prefixes. A simple example of languages $L_1=a(ba)^*$ and $L_2=b(ab)^*$ shows that the towers of prefixes and towers (of subsequences) may behave differently. Indeed, there is no infinite tower of prefixes between $L_1$ and $L_2$, since every word of $L_1$ begins with $a$ and cannot thus be a prefix of a word of $L_2$, which begins with $b$. But there is an infinite tower, namely, $a, bab, ababa, \ldots$.

  We first describe a pattern on two automata $\A$ and $\B$ that characterizes the existence of an infinite tower of prefixes between them. 
  \begin{figure}[bh]
    \centering
    \begin{tikzpicture}[->,>=stealth,shorten >=1pt,auto,node distance=1.1cm,
    state/.style={circle,minimum size=6mm,very thin,draw=black,initial text=}]
    \node[]       (1) {};
    \node[state]  (2) [left of=1]  {$\sigma$};
    \node[state]  (3) [right of=1] {$\tau$};
    \node[state]  (4) [below of=1]  {$\sigma_2$};
    \node[state]  (5) [above of=1] {$\tau_2$};
    \path
    (2) edge[-stealth,decoration={snake,amplitude=.4mm,segment length=2mm,post length=1.9mm},decorate,bend right] node {$x$} (4)
    (4) edge[-stealth,decoration={snake,amplitude=.4mm,segment length=2mm,post length=1.9mm},decorate,bend right] node {$u_1$} (3)
    (3) edge[-stealth,decoration={snake,amplitude=.4mm,segment length=2mm,post length=1.9mm},decorate,bend right] node {$y$} (5)
    (5) edge[-stealth,decoration={snake,amplitude=.4mm,segment length=2mm,post length=1.9mm},decorate,bend right] node {$u_2$} (2) ;
    \begin{pgfonlayer}{background}
    \filldraw [line width=4mm,black!10] (1 -| 1) ellipse (1.4cm and 1.4cm);
    \end{pgfonlayer}
    \node[] (6) [left of=2]  {};
    \node[] (7) [right of=3] {};
    \node[] (8) [below left of=6]  {};
    \node[state,semicircle,draw,semicircle,rotate=90,scale=.5,fill=black!10,anchor=south] at (8) {};
    \node[state,semicircle,draw,semicircle,rotate=-90,scale=.5,anchor=south] at (8) {};
    \node[] (8) [below left of=6]  {$\sigma_1$};
    \node[] (9) [below right of=7] {};
    \node[state,semicircle,draw,semicircle,rotate=-90,scale=.5,fill=black!10,anchor=south] at (9) {};
    \node[state,semicircle,draw,semicircle,rotate=90,scale=.5,anchor=south] at (9) {};
    \node[] (9) [below right of=7] {$\tau_1$};
    \node[state,initial] (i) [above=of 8] {};
    \path
    (2) edge[-stealth,decoration={snake,amplitude=.4mm,segment length=2mm,post length=1.9mm},decorate] node[above] {$x$} (8)
    (3) edge[-stealth,decoration={snake,amplitude=.4mm,segment length=2mm,post length=1.9mm},decorate] node{$y$} (9)
    (i) edge[-stealth,decoration={snake,amplitude=.4mm,segment length=2mm,post length=1.9mm},decorate] node[above]{$u$} (2) ;
    \end{tikzpicture}
    \caption{The pattern $(\sigma,\sigma_1,\sigma_2,\tau,\tau_1,\tau_2)$}
    \label{fig1}
  \end{figure}
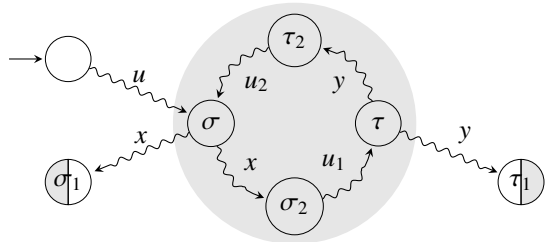
  Let $\A=(Q_A,\Sigma,\delta_A,q_A,F_A)$ and $\B=(Q_B,\Sigma,\delta_B,q_B,F_B)$ be two NFAs. We say that $(\sigma,\sigma_1,\sigma_2,\tau,\tau_1,\tau_2)$ is a \emph{pattern} of the automata $\A$ and $\B$ if $\sigma$, $\sigma_1$, $\sigma_2$, $\tau$, $\tau_1$, $\tau_2$ are states of the product automaton such that 
  \begin{itemize}
    \itemsep0pt
    \item $\sigma_1\in F_A\times Q_B$ and $\tau_1\in Q_A\times F_B$,
    \item $\sigma$ is reachable from the initial state,
%	\item $\sigma, \sigma_2, \tau, \tau_2$ belong to $S$,
    \item states $\sigma_1$ and $\sigma_2$ are reachable from state $\sigma$ under a common word,
    \item states $\tau_1$ and $\tau_2$ are reachable from state $\tau$ under a common word, and
%	\item the strongly connected component $S$ is reachable from the initial state.
    \item $\tau$ is reachable from $\sigma_2$ and $\sigma$ is reachable from $\tau_2$.
  \end{itemize}
  The definition is illustrated in Figure~\ref{fig1}. We allow any of the words in the definition to be empty, with the convention that any state is reachable from itself under the empty word. 

  The following theorem provides a characterization for the existence of an infinite tower of prefixes.
  \begin{thm}\label{patern}
    Let $\A$ and $\B$ be two NFAs. Then there is an infinite tower of prefixes between $\A$ and $\B$ if and only if there is a pattern of automata $\A$ and $\B$.
  \end{thm}
  \begin{proof}
    Let $(\sigma,\sigma_1,\sigma_2,\tau,\tau_1,\tau_2)$ be a pattern of the automata $\A$ and $\B$. Let $u$ denote a word under which state $\sigma$ is reachable from the initial state $(q_A,q_B)$. Let $x$ ($y$ resp.) be a word under which both $\sigma_1$ and $\sigma_2$ ($\tau_1$ and $\tau_2$ resp.) are reachable from $\sigma$ ($\tau$ resp.). Let $u_1$ denote a word under which $\tau$ is reachable from $\sigma_2$, and $u_2$ a word under which $\sigma$ is reachable from $\tau_2$, see Figure~\ref{fig1}. Together, we have an infinite tower of prefixes $ux$, $ux(u_1y)$, $ux(u_1y)(u_2x), ux(u_1y)(u_2x)(u_1y),\dots$. 
    %$u (x u_1 y u_2)^* (x + x u_1 y)$.

    Assume now that there exists an infinite tower of prefixes $(w_i)_{i=1}^{\infty}$ between the languages $L(\A)$ and $L(\B)$. Consider the automaton $\det(\A\times \B)$, the determinization of $\A\times \B$ by the standard subset construction, and let $q_{\A\times\B}$ be its initial state. A sufficiently long element of the tower defines a path \[
      q_{\A\times\B}\xrightarrow{u} X \xrightarrow{z_X} Y \xrightarrow{z_Y} X
    \] 
    in the automaton $\det(\A\times \B)$, such that $X$ contains a state $(f_1,q_1)\in F_A\times Q_B$ and $Y$ contains a state $(q_2,f_2)\in Q_A\times F_B$. 
		For every state of $X$, there exists an incoming path from an element of $Y$ labeled by $z_Y$ since $X=\delta_{\A\times \B}(Y,z_Y)$. Similarly, for every state of $Y$, there exists an incoming path from an element of $X$ labeled by $z_X$ since  $Y=\delta_{\A\times \B}(X,z_X)$. Thus, there are infinitely many paths from $X$ to $X$ labeled with words from $(z_Xz_Y)^+$ ending in state $(f_1,q_1)$. Therefore, there exists a state $(s_1,t_1)\in X$ and integers $k_1$ and $\ell_1$ such that 
    \[
      (s_1,t_1) \xrightarrow{(z_Xz_Y)^{k_1}}    (s_1,t_1) 
                \xrightarrow{(z_Xz_Y)^{\ell_1}} (f_1,q_1)\,.
    \] 
    Similarly, there exists a state $(s_2,t_2)\in X$ and integers $k_2$ and $\ell_2$ such that 
    \[
      (s_2,t_2) \xrightarrow{(z_Xz_Y)^{k_2}}        (s_2,t_2)
                \xrightarrow{(z_Xz_Y)^{\ell_2}z_X}  (q_2,f_2)\,.
    \]
	
    Let $\sigma = \tau = (s_1,t_2)$. Since $(q_{\A\times\B})\xrightarrow{u} (s_i,t_i)$, $i=1,2$, also $(q_{\A\times\B})\xrightarrow{u} \sigma$. Let $x=(z_Xz_Y)^{\ell_1}$ and $y=(z_Xz_Y)^{\ell_2}z_X$. Then $\sigma \xrightarrow{x} (f_1,t_3)$ where $t_3$ is a state in the cycle $t_2\xrightarrow{(z_Xz_Y)^{k_2}} t_2$ in $\B$. Similarly, $\tau \xrightarrow{y} (s_3,f_2)$ where $s_3$ is a state in the cycle $s_1\xrightarrow{(z_Xz_Y)^{k_1}} s_1$ in $\A$.  We set $\sigma_1=(f_1,t_3)$ and $\tau_1=(s_3,f_2)$.
    The pattern is completed by states $\sigma_2$ and $\tau_2$,  such that
	 \[\sigma \xrightarrow{x} \sigma_2 \xrightarrow{u_1} \tau \xrightarrow{y} \tau_2 \xrightarrow{u_2} \sigma
	  \]
	 where $u_1$ and $u_2$ can be chosen as $u_1=(z_Xz_Y)^{\ell_1k_1k_2-\ell_1}$ and $u_2=z_Y(z_Xz_Y)^{\ell_2k_1k_2-\ell_2-1}$. 
\end{proof}

We point out that the pattern can easily be identified. It could even be shown that to decide whether there is a pattern between the automata, that is, whether there is an infinite tower of prefixes, is an NL-complete problem for both NFAs and DFAs. This is in contrast to deciding the existence of an infinite tower of subsequences, which is PTime-complete~\cite{tm2016}.

  We have already mentioned that if there are towers of arbitrary height, then there is an infinite tower. This property holds for any relation that is a well quasi order (WQO) ~\cite[Lemma~6]{icalp2013} of which the subsequence relation is an instance. The prefix relation is not a WQO. However, Theorem \ref{patern} and its proof shows that the pattern and therefore also an infinite tower of prefixes can be found as soon as there exists a sufficiently long tower of prefixes. On the other hand, this argument depends on the fact that the languages are regular. Indeed, the following example shows that the property in general does not hold for non-regular languages.

  \begin{example}\label{nonregular}
    Let $K=\{a,b\}^*a$ and $L=\{ a^m (ba^*)^nb \mid m > n \ge 0\}$ be two languages. Note that $K$ is regular and $L$ is non-regular context-free. The languages are disjoint, since the words of $K$ end with $a$ and the words of $L$ with $b$.
    For any $k\geq 1$, the words $w_{2i+1}=a^k(ba)^i\in K$ and $w_{2(i+1)}=a^k(ba)^ib\in L$, for $i=0,1,\dots,k-1$, form a tower of prefixes between $K$ and $L$ of height $2k$.
    On the other hand, let $w_1, w_2, \ldots$ be a tower of prefixes between the languages $K$ and $L$. Without loss of generality, we may assume that $w_1$ belongs to $L$. Then $a^kb$ is a prefix of $w_1$, for some $k\ge 1$. It is not hard to see that $|w_i|_b<|w_{i+2}|_b$ holds for any $w_i\le w_{i+1} \le w_{i+2}$ with $w_i, w_{i+2}$ in $L$ and $w_{i+1}$ in $K$. As any word of $L$ with a prefix $a^kb$ can have at most $k$ occurrences of letter $b$, the tower cannot be infinite.
  \end{example}

Given that the height of finite towers of prefixes for regular languages is bounded, we now investigate the bound.  We shall need the following auxiliary lemma. 
	
	\begin{lem}\label{ab}
    Let $k_1,\ell_1,k_2,\ell_2\geq 0$ be integers and $k_1+k_2>0$ and $\ell_1+\ell_2>0$. Then
    $
      2\cdot\min(k_1k_2,\ell_1\ell_2)\leq \frac {(k_1+\ell_1)(k_2+\ell_2)}{2}\,.
    $
    Moreover, if $k_1k_2\neq \ell_1\ell_2$, then 
    $
      2\cdot\min(k_1k_2,\ell_1\ell_2)+1\leq \frac {(k_1+\ell_1)(k_2+\ell_2)}{2}\,.
    $
	\end{lem}
  \begin{proof} 
    Suppose that $k_1=0$. Then the first claim is obvious. The condition of the second one holds only if $\ell_1,\ell_2 >0$. Since also $k_2>0$, we get the second claim. By symmetry, we shall further suppose that $k_1,\ell_1,k_2,\ell_2\geq 1$.
    
    Let now $k_1\leq \ell_1$ and $k_2<\ell_2$.
    Then
    $
      2\min(k_1k_2,\ell_1\ell_2)+1=2k_1k_2 + 1 \leq 2k_1k_2 + k_1  = \frac {2k_1(2k_2+1)}{2} \leq \frac {(k_1+\ell_1)(k_2+\ell_2)}{2}\,.
    $

    By symmetry, it remains to consider the case $k_1\leq \ell_1$ and $\ell_2\leq k_2$. Set $d_1=\ell_1-k_1$ and $d_2=k_2-\ell_2$. By symmetry, we may also suppose $d_1\ell_2\leq d_2k_1$. Then 
    \begin{align*}
      2\min(k_1k_2,\ell_1\ell_2)
      &=2\min(k_1(\ell_2+d_2),(k_1+d_1)\ell_2)=2k_1\ell_2+2d_1\ell_2\\
      &\leq 2k_1\ell_2+d_1\ell_2 +d_2k_1 +\frac{d_1d_2}2 = \frac {(2k_1+d_1)(2\ell_2+d_2)}{2} 
        = \frac {(k_1+\ell_1)(k_2+\ell_2)}{2}\,.
    \end{align*}
    If $k_1k_2 =k_1(\ell_2+d_2)\neq (k_1+d_1)\ell_2=\ell_1\ell_2$, then $d_1\ell_2< d_2k_1$, that is, $d_1\ell_2+1\leq d_2k_1$, and the obvious modification of the last formula yields the second claim.
  \end{proof}
  
For DFAs we have the following bound.
  \begin{thm}\label{thm:dfas}
    Let $\A$ and $\B$ be two nonempty DFAs with $n_1$ and $n_2$ states that have no infinite tower of prefixes. Then the height of a tower of prefixes between $\A$ and $\B$ is at most $\frac{n_1n_2}{2} + 1$.
  \end{thm}
  \begin{proof}
    Let $\A=(Q_{\A},\Sigma,\delta_{\A},q_{\A},F_{\A})$ and $\B=(Q_{\B},\Sigma,\delta_{\B}, q_{\B}, F_{\B})$, and let $X=F_{\A}\times (Q_{\B}\setminus F_{\B})$ and $Y=(Q_{\A}\setminus F_{\A})\times F_{\B}$. 
		Final states $(p_i,q_i)=\delta((q_{\A},q_{\B}),w_i)$ of any tower of prefixes $(w_i)_{i=1}^{r}$ between $\A$ and $\B$ in the product automaton $\A\times \B$ have to alternate between the states of $X$ and $Y$, with the exception of $w_r$: there may be no path labeled by $w_r$ in the non-accepting automaton, and therefore also no path in the product automaton $\A\times \B$ (recall our convention not to consider states that do not appear on an accepting path). 

    If $(p_i,q_i)=(p_j,q_j)$ for some $1\leq i<j<r$, then there is a path
    \[
      (q_{\A},q_{\B}) \xrightarrow{w_i} (p_i,q_i) 
                      \xrightarrow{u}   (p_{i+1},q_{i+1})
                      \xrightarrow{v}   (p_i,q_i)\,.
    \]    
    with $w_{i+1}=w_iu$ and $w_j=w_iuv$. Then there is an infinite tower of prefixes $w_i, w_iu, w_iuv, w_iuvu, \dots$, a contradiction. 
    Therefore, it remains to show that there may be at most $\frac{n_1n_2}2$ alternations without repeated states between $X$ and $Y$.
    
    If $|X|=|Y|$, then there are at most $2\min(|X|,|Y|)$ such alternations.
    If $|X|\neq |Y|$, then there are at most $2\min(|X|,|Y|)+1$ such alternations. In both cases, the proof is completed by Lemma \ref{ab} with $k_1=\abs{F_{\A}}$, $\ell_2=\abs{F_{\B}}$, and $n_i=k_i+\ell_i$, $i=1,2$, noting that for $k_1=0$ or $\ell_2=0$ the claim holds since then $L(\A)$ or $L(\B)$ is empty.
  \end{proof}

  The following theorem allows to conclude that the above bound is tight.
  \begin{thm}\label{thm:dfas:tight}
    For every positive integer $d$ and every odd positive integer $e$, there exists a binary DFA with $2d$ states and a binary DFA with $e+1$ states having a tower of prefixes of height $d(e+1)+1$ and no infinite tower.
  \end{thm}
    \begin{figure}
      \centering
      \begin{tikzpicture}[->,>=stealth,shorten >=1pt,node distance=1.6cm,
        state/.style={circle,minimum size=1mm,very thin,draw=black,initial text=},
        every node/.style={fill=white,font=\small},
        bigloop/.style={shift={(0,0.01)},text width=1.6cm,align=center},
        bigloopd/.style={shift={(0,-0.01)},text width=1.6cm,align=center}]
        \node[state,initial,accepting]  (1) {$1$};
        \node[state]                    (1') [below of=1] {$1'$};
        \node[state,accepting]          (2) [right of=1] {$2$};
        \node[state]                    (2') [below of=2] {$2'$};
        \node[state,accepting]          (3) [right of=2] {$3$};
        \node[state]                    (3') [below of=3] {$3'$};
        \node[state,accepting]          (4) [right of=3] {$4$};
        \node[state]                    (4') [below of=4] {$4'$};
        \node[state,accepting]          (5) [right of=4] {$5$};
        \node[state]                    (5') [below of=5] {$5'$};
        \path
          (1') edge node {$a$} (2)
          (2') edge node {$a$} (3)
          (3') edge node {$a$} (4)
          (4') edge node {$a$} (5)
          (1) edge node {$a$} (2)
          (2) edge node {$a$} (3)
          (3) edge node {$a$} (4)
          (4) edge node {$a$} (5)
        ;
        \foreach \from/\to in {1/1',2/2',3/3',4/4',5/5', 1'/1,2'/2,3'/3,4'/4,5'/5}
          \path (\from) edge[bend left] node {$b$} (\to);
      \end{tikzpicture}
      \quad\quad
      \begin{tikzpicture}[->,>=stealth,shorten >=1pt,node distance=1.6cm,
        state/.style={circle,minimum size=1mm,very thin,draw=black,initial text=},
        every node/.style={fill=white,font=\small},
        bigloop/.style={shift={(0,0.01)},text width=1.6cm,align=center}]
        \node[state,initial]    (1) {$0$};
        \node[state,accepting]  (2) [above right of=1]{$1$};
        \node[state]            (3) [right of=2] {$2$};
        \node[state, accepting] (4) [below right of=3] {$3$};
        \node[state]            (5) [below left of=4] {$4$};
        \node[state, accepting] (6) [left of=5] {$5$};
        
        \path
          (1) edge[bend left=15] node {$b$} (2)
          (2) edge[bend left=15] node {$b$} (3)
          (3) edge[bend left=15] node {$b$} (4)
          (4) edge[bend left=15] node {$b$} (5)
          (5) edge[bend left=15] node {$b$} (6)
          (6) edge[bend left=15] node {$a$} (1);
       \end{tikzpicture}
      \caption{DFAs $\A_d$ and $\B_e$ of Theorem~\ref{thm:dfas:tight} for $d=e=5$}
      \label{exdfa}
    \end{figure}
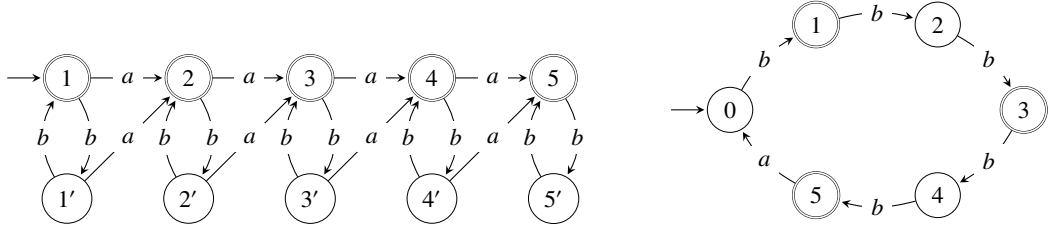
   \begin{proof}
    We consider the proof of Theorem~\ref{thm02B}, but instead of taking the NFA $\A_d$, we take its DFA equivalent, which has $2d$ states and, for simplicity, we denote it $\A_d$ as well, cf. Figure~\ref{exdfa}. From Theorem~\ref{thm02B}, there is no infinite tower between the languages, hence there is also no infinite tower of prefixes between the DFAs.
    
    Consider the word $w=(b^e a)^{d-1} b^{e+1}$. By the proof of Theorem~\ref{thm02B}, $\A_d$ accepts all prefixes of $w$ ending with an even number of $b$'s, including those ending with $a$, and $\B_e$ accepts all prefixes of $w$ ending with an odd number of $b$'s. The sequence $(w_i)_{i=1}^{|w|+1}$, where $w_i$ is the prefix of $w$ of length $i-1$, for $i=1,2,\dots,|w|+1$, is therefore a tower of prefixes between $\A_d$ and $\B_e$ of height $|w|+1 = (e+1)(d-1) + e + 1 + 1 = d(e+1)+1$. (The last word of the tower in Theorem~\ref{thm02B} does not fit to a prefix tower.)
  \end{proof}

  \begin{cor}\label{cor:dfas}
    For every even positive integers $n_1$ and $n_2$, there exist binary DFAs with $n_1$ and $n_2$ states having a tower of prefixes of height $\frac{n_1n_2}{2}+1$ and no infinite tower.
  \end{cor}
  \begin{proof}
    The proof follows from Theorem~\ref{thm:dfas:tight} by setting $d=\frac{n_1}{2}$ and $e=n_2-1$.
  \end{proof}

  Comparing towers of subsequences and prefixes with respect to the number of states of DFAs, Theorem~\ref{thm:expdfa} shows that there are towers of subsequences of exponential height, while Theorem~\ref{thm:dfas} gives a quadratic bound on the height of towers of prefixes. It shows an exponential gap between the height of towers of subsequences and prefixes for DFAs.
  What is the situation for NFAs? An immediate consequence of the NFA-to-DFA transformation and Theorem~\ref{thm:dfas} give the following asymptotically tight bound.

  %\textbf{The following corollary is strange: the lower bound is not symmetrc in $n_1$ and $n_2$.}
  \begin{cor}\label{cor:nfas}
    Given two NFAs with at most $n_1$ and $n_2$ states and with no infinite tower of prefixes, the height of a tower of prefixes between them is at most $2^{n_1+n_2-1}-2^{n_1-1}-2^{n_2-1}+1$. Moreover, the lower bound is $2^{n_1+n_2-2}-2^{n_2-2}+1$ for any $n_1,n_2\geq 2$.
  \end{cor}
  \begin{proof}
    Let two NFAs with $n_1$ and $n_2$ states. Their corresponding minimal DFAs have at most $2^{n_1}-1$ and $2^{n_2}-1$ nonempty states. By Theorem~\ref{thm:dfas}, the upper bound on the height of towers of prefixes is $\frac12{(2^{n_1}-1)(2^{n_2}-1)}+1$. Taking the integer part, the height is at most 
		$
      \frac{(2^{n_1}-1)(2^{n_2}-1)+1}{2} = 2^{n_1+n_2-1}-2^{n_1-1}-2^{n_2-1}+1.
    $
    
		The lower bound is obtained from Corollary \ref{cor:A1} noting that the tower constructed in the proof of  Theorem \ref{thm:2exp:improved} is a tower of prefixes.
  \end{proof}

  A natural question is whether there are any requirements on the size of the alphabet in case of automata with exponentially high towers of prefixes. The following corollary shows that the alphabet can be binary and the tower is still more than polynomial in the number of states.

  \begin{cor}\label{cor23}
    For any $n$ there are binary NFAs with at most $n$ states with no infinite tower of prefixes and with a tower of prefixes of a superpolynomial height with respect to $n$.
  \end{cor}
  \begin{proof}
    The property of being a tower of prefixes is preserved if the alphabet is encoded in binary. The binary code of each letter has length at most $\log m$ for an alphabet of cardinality $m$. Therefore, every transition under an original letter can be replaced by a path with at most $\log m$ new states. 

    Consider the automata $\A_{m_A,m_B,(1,\dots,1)}$ and $\B_{m_A,m_B,(1,\dots,1)}$ of Theorem \ref{thm:2exp:improved}. They have $m_A+1$ and $m_B+2$ states and $O((m_A+m_B)^2)$ transitions. Encoding every letter in binary results in automata with $n_1+n_2=O((m_A+m_B)^2\log(m_A+m_B))$ states and a tower of height at least 
    $2^{m_A+m_B}\in2^{\Omega\left(\sqrt{\frac {n_1+n_2}{\log (n_1+n_2)}}\right)}$.
  \end{proof}

  The following question is open.
  \begin{oprob}\label{op2}
    Given two NFAs with $n_1$ and $n_2$ states over a fixed alphabet with $m$ letters. Assume that there is no infinite tower of prefixes between the automata. What is the tight bound on the height of towers of prefixes?
  \end{oprob}

  \section{Conclusion}
  We investigated the height of finite towers between two regular languages as a function of the number of states of the automata representing the languages. We also paid attention to three additional parameters: (non)determinism, the size of the alphabet, and the structure of the tower (formed by subsequences or by prefixes). The connection between the parameters is summarized as follows (for an overview of the results see Table~\ref{tableOverview}).

  The NFA vs. DFA representation does not play a crucial role since any tower between two NFAs can be ``determinized'' to a tower between two DFAs with only a moderate increase of the number of states. 
  
  A difference between towers of subsequences and towers of prefixes is less clear. It is conspicuous that our best, exponentially high towers are essentially towers of prefixes. Although this holds only for NFAs (for DFAs and towers of prefixes we have achieved an exact quadratic bound), it is worth noting that the proper subsequence relation is used exclusively in the determinization constructions. It leaves an intriguing open question whether, in the nondeterministic case, there is any substantial difference between towers of subsequences and towers of prefixes. In other words, the question is whether the subsequence relation can be simulated by the prefix relation using nondeterminism.

  Unclear is also the real influence of the alphabet size. We have seen that the height of towers grows exponentially with the alphabet size up to the point when the alphabet size is roughly the same as the number of states. The second intriguing question is whether the towers can grow with the alphabet beyond this point. The unconditional upper bound we have obtained is $O\left(n^{|\Sigma|}\right)$, where the only limit on the size of $\Sigma$ is the trivial bound $2^{n^2}$ on the number of inequivalent letters (a letter $a$ can be identified with the mapping $\delta(\,\cdot\,,a): Q \to 2^Q$). 
  
  The two open questions are related. If, for NFAs, towers of prefixes are as high as towers of subsequences, then $\Theta(2^{n_1+n_2})$ is the optimal bound (cf. Open Problem~\ref{op1}).

\end{document}